\documentclass{article}
\usepackage{praneeth}
\usepackage{tikz}
\usepackage{sourceserifpro}
\usepackage[colorlinks]{hyperref}
\usepackage[algo2e,ruled,linesnumbered]{algorithm2e}

\usepackage[square, numbers, sort]{natbib}
\author{Nadiia Chepurko\footnote{MIT. Email: \href{mailto:nadiia@mit.edu}{\texttt{nadiia@mit.edu}}} \and Kenneth L. Clarkson\footnote{IBM Research Almaden. Email: \href{mailto:klclarks@us.ibm.com}{\texttt{klclarks@us.ibm.com}}} \and Praneeth Kacham\footnote{CMU. Email: \href{mailto:pkacham@cs.cmu.edu}{\texttt{pkacham@cs.cmu.edu}}} \and David P. Woodruff\footnote{CMU. Email: \href{mailto:dwoodruf@cs.cmu.edu}{\texttt{dwoodruf@cs.cmu.edu}}}}
\date{}
\title{Near-Optimal Algorithms for Linear Algebra in the Current Matrix Multiplication Time}

\newcommand{\lev}{\textnormal{lev}}
\newcommand{\fast}{\textnormal{fast}}
\newcommand{\OPT}{\textnormal{OPT}}
\renewcommand{\T}[1]{#1^\mathsf{T}}
\renewcommand{\epsilon}{\varepsilon}
\begin{document}
\maketitle
\begin{abstract}
In the numerical linear algebra community, it was suggested that to obtain nearly optimal bounds for various problems such as rank computation, finding a maximal linearly independent subset of columns (a \emph{basis}), regression, or low-rank approximation, a natural way would be to resolve the main open question of Nelson and Nguyen (FOCS, 2013). This question is regarding the logarithmic factors in the sketching dimension of existing  oblivious subspace embeddings that achieve constant-factor approximation.
We show how to bypass this question using a refined sketching technique, and obtain optimal or nearly optimal bounds for these problems.
A key technique we use is an explicit mapping of Indyk based on uncertainty principles and extractors, which after first applying known oblivious subspace embeddings, allows us to quickly spread out the mass of the vector so that sampling is now effective.
We thereby avoid a logarithmic factor in the sketching dimension that is standard in bounds proven using the matrix Chernoff inequality.   
For the fundamental problems of rank computation and finding a basis, our algorithms improve Cheung, Kwok, and Lau (JACM, 2013), and are optimal to within a constant factor and a $\poly(\log\log(n))$-factor, respectively. Further, for constant-factor regression and low-rank approximation we give the first optimal algorithms, for the current matrix multiplication exponent.
\end{abstract}


\section{Introduction}
We obtain several new results for fundamental problems in numerical linear algebra, in many cases removing,
in particular, the \emph{last} log factor to obtain a running time that is truly linear in the input sparsity, and with lower-order terms that are close to optimal.
We note that the bottleneck in improving prior work, including such removal of
last logarithmic factors, involved well-known conjectures to construct \textit{Sparse Johnson-Lindenstrauss transforms} (see Conjecture 14 in \cite{NN13}).

To sidestep these conjectures 
we introduce a new simple matrix sketching
technique which allows for  multiplication by a random
sparse matrix
whose randomly chosen nonzero entries are random signs. The key idea is to compose 
this matrix with an appropriate Flattening transform based on explicit embeddings of $\ell_2$ into $\ell_1$, together with \textsf{OSNAP} embeddings. Using this, we obtain the first oblivious subspace embedding for $k$-dimensional subspaces that has $o(k\log(k))$ rows and that can be applied to a matrix $A$ in time asymptotically less than both $\nnz(A) \log k$ and $k^{\omega} \log k$, where $\nnz(A)$ is the number of nonzero entries in the matrix $A$, and $\omega \approx 2.37$ is the exponent of fast matrix multiplication \cite{alman2021refined}. 
This scheme removes a log factor that has thus far remained both a nuisance and an impediment to
optimal algorithms. Our main embedding result is as follows:

\begin{theorem}[Fast Subspace Embedding, informal Theorem~\ref{thm:sparse-subspace-embedding}]
Given an $n \times k$ matrix, there is a distribution $\calS$ over matrices with $k\poly(\log\log k)$ rows such that, for $\bS \sim \calS$, with probability $\ge 99/100$, for all vectors $x \in \R^{k}$
\begin{equation*}
    \opnorm{Ax} \le \opnorm{\bS Ax} \le \exp(\poly(\log\log k))\opnorm{Ax}.
\end{equation*}
For $\bS \sim \calS$, with probability $\ge 95/100$, the matrix $\bS A$ can be computed in time $O(\gamma^{-1}\nnz(A) + k^{2+\gamma+o(1)})$ for any constant $\gamma > 0$.
\end{theorem}

Using our subspace embedding, together with additional ideas, we obtain nearly optimal (up to $\log\log$ factors in the sub-linear terms) running times for fundamental problems in classical linear algebra including computing matrix rank, finding a set of linearly independent rows, and linear regression. Further, for regression and low-rank approximation, we obtain the first optimal algorithms for the current matrix multiplication exponent. We begin with least-squares regression:

\begin{theorem}[Least-Squares Regression, informal Theorem \ref{thm:linear-regression}]
Given a full rank $n \times k$ matrix $A$, $k \le n$, and vector~$b$, there exists an algorithm that computes $\hat{x}$ such that $\|A\hat{x} - b \|_2 \leq (1+\eps)\min_{x} \|A x - b \|_2$ in time 
\begin{equation*}
  O\left(\frac{\nnz(A)}{\gamma} + k^{\omega}\poly(\log\log(k)) +\frac{1}{ \poly(\varepsilon)}k^{2+o(1)}n^{\gamma+o(1)}\right)  
\end{equation*}
for any constant $\gamma > 0$ small enough. 
\end{theorem}

We note that for constant $\eps$ and $k = n^{\Omega(1)}$, the running time obtained is within a $\poly(\log\log(n))$ factor of  optimal, for the current matrix multiplication constant. Further, it improves on prior work \cite{CW13,mm13,NN13,bourgain2015toward,cohen2015uniform,cohen2015optimal} describing algorithms with an additional $\log(n)$ factor multiplying either the leading $\nnz(A)$ term, or that is $\nnz(A)$ time but has a $k^{\omega} \log k$ additive term or worse. We note that our additive term is only $k^{\omega}\poly(\log \log k)$, for the current matrix multiplication exponent $\omega$, when $k = n^{\Omega(1)}$.  
Importantly, up to a $\poly(\log \log k)$ factor, our bound is best possible, and thus we remove the last logarithmic factor even in the additive term. As we explain more below, the issue with previous work is that to obtain a sketching dimension of $O(k)$, for constant $\epsilon$, one needs either $\nnz(A) k$ time to directly perform a multiplication with a dense Sub-Gaussian matrix, or at least $k^{\omega}\log k$ time to compose a dense Sub-Gaussian sketch with a sparse sketch. We avoid this using our new subspace embedding, given by   Theorem~\ref{thm:sparse-subspace-embedding}. 

We note that simply sketching on the left with a CountSketch matrix and solving the sketched problem attains an optimal $O(\nnz(A))$ running time for $k = O(n^c)$ for a sufficiently small constant $c > 0$,  and so our theorems are 
most interesting when $k = \Omega(n^c)$. 

Next, we show a similar result holds for low-rank approximation (LRA): 

\begin{theorem}[LRA in Current Matrix Multiplication Time, informal Theorem \ref{thm:final-low-rank-approximation}] 
Given $\eps > 0$, an $n \times d$ matrix $A$ and $k \le \min(n,d)$, $k = \max(n,d)^{\Omega(1)}$, there exists an algorithm that runs in $$O\left(\nnz(A) + \frac{(n+d)k^{\omega-1}}{\eps} + \frac{(n+d)k^{1.01}}{\eps} + \poly(\varepsilon^{-1}k)\right)$$ time and outputs two matrices $V \in \R^{n \times k}$ and $\tilde{X} \in \R^{k \times d}$, with $\T{V}V = I_k$, such that
\begin{equation*}
	\frnorm{A - V \cdot \tilde{X}} \le (1+\varepsilon)\frnorm{A-[A]_k}.
\end{equation*}
\end{theorem}
For the current matrix multiplication exponent, the running time is $O(\nnz(A) + (n+d)k^{\omega-1})$ for constant $\varepsilon$. In contrast, existing low rank approximation algorithms
\cite{CW13,mm13,NN13,bourgain2015toward,cohen2015dimensionality,cohen2015uniform,cohen2015optimal,cohen2017input} take time at least $\nnz(A) \log n$ or $d k^{\omega-1} \log k$ or worse. 
Thus, as with least squares regression, we remove the last logarithmic factor in both the $\nnz(A)$ term and the 
leading additive term. 

We also give constructions of $1+\varepsilon$ subspace embeddings with $O(k\log(k)/\varepsilon^2)$ rows that have better running times than earlier subspace embeddings with $O(k\log(k)/\varepsilon^2)$ rows, such as approximate leverage score sampling and \textsf{OSNAP} embeddings.

\begin{theorem}[Subspace Embeddings, informal Theorem~\ref{thm:final-epsilon-subspace-embedding}] Given a matrix $A \in \R^{n \times k}$, there is a \emph{non-oblivious} subspace embedding $\bS$ with $O(k\log(k)/\varepsilon^2)$ rows that can be applied to the matrix $A$ in time
	$O(\nnz(A) + k^{\omega}\poly(\log\log k) + \poly(\varepsilon^{-1})k^{2.1+o(1)})$ for $k = n^{\Omega(1)}$.
\end{theorem}

%

Finally, we obtain faster algorithms for computing the rank of a matrix and finding a full-rank set of rows.  

\begin{theorem}[Matrix Rank and Finding a Basis, informal Theorem \ref{thm rank est} and \ref{thm:independent_rows}]
Given an $n \times d$ matrix $A$, there exists a randomized algorithm to compute $k = \rank(A)$ in $O(\nnz(A) + k^\omega)$ time, where $\omega$ is the matrix multiplication constant. Further, the algorithm can find a set of $k$ linearly independent rows in $O(\nnz(A)+ k^\omega \log\log(n))$ time.
\end{theorem}

We note that this result improves prior work by \citet{CKL13}, in the case of matrices with real numbers, who obtain an $O(\nnz(A)\log(k) + k^\omega )$ time algorithm to compute matrix rank and an $O( \log(n)(\nnz(A) + k^\omega))$ time algorithm to find a full-rank set of rows. 

The following table lists our running times for $k\le n$ and $k = n^{\Omega(1)}$, assuming $\omega > 2$,
and putting some terms to constant values (such as 2.1 instead of $2+\gamma$). See theorem statements for exact running times.
\begin{center}
\begin{tabular}{l l}\hline
	Application & Running time (up to constant factors)\\ \hline
	$\varepsilon$ Subspace Embeddings & $\nnz(A) + \varepsilon^{-3}k^{2.1+o(1)} + k^{\omega}\poly(\log\log(k))$\\
	$\varepsilon$ approximate linear regression & $\nnz(A) + \varepsilon^{-3}k^{2.1+o(1)} + k^{\omega}\poly(\log\log(k))$\\
	Linearly Independent Rows & $\nnz(A) + k^\omega \poly(\log\log(k)) + k^{2+o(1)}$\\
	$0.01$ Low-Rank Approximation & $\nnz(A) + (n+d)k^{\omega-1}$\\
	\hline
\end{tabular}
\end{center}
\section{Related Work}

\paragraph{Matrix Sketching.}
The \textit{sketch and solve} paradigm \cite{clarkson2015input,woodruff2014sketching} was designed to reduce the dimensionality of a problem, while maintaining enough structure such that a solution to the smaller problem remains an approximate solution to the original one. This approach has been pivotal in speeding up basic linear algebra primitives such as least-squares regression \cite{s06, rokhlin2008fast, clarkson2015input}, $\ell_p$ regression \cite{cohen2015lewis, wang2019tight}, low-rank approximation \cite{NN13, cohen2017input, li2020input}, linear and semi-definite programming \cite{cohen2019solving,jiang2020inv, jiang2020faster}, solving non-convex optimization problems such as $\ell_p$ low-rank approximation \cite{song2017low,song2019relative,ban2019ptas}, and training neural networks \cite{bakshi2019learning,brand2020training}. For a comprehensive overview we refer the reader to the aforementioned papers and citations therein.  Several applications use rank computation, finding a full rank subset of rows/columns, leverage score sampling, and computing subspace embeddings, as key algorithmic primitives. In addition to being used as a black box, we believe our techniques will be useful in sharpening bounds for several such applications.

\section{Preliminaries}
\label{subsec nota}
\paragraph{Computational Model} Throughout the paper, we work with matrices having real numbers and assume that all elementary arithmetic operations on real numbers can be computed in $O(1)$ time.

Let $A^+$ denote the Moore-Penrose pseudo-inverse of matrix $A \in \R^{n \times d}$, equal to $V\Sigma^{-1}U^\top$ when
$A$ has ``thin'' Singular Value Decomposition (SVD) $A=U\Sigma V^\top$, so that $\Sigma$ is a square invertible matrix. 
We note that $AA^+$ is the projection matrix onto the column span of the matrix $A$. 
Let $\opnorm{A}$ denote the spectral norm ($\ell_2 \rightarrow \ell_2$ operator norm)  of $A$ and $\frnorm{A}$ denote the Frobenius norm $(\sum_{i,j}A_{ij}^2)^{1/2}$.
Let $\kappa(A)=\opnorm{A^+}\opnorm{A}$ denote the condition number of $A$. 
We write $a\pm b$ to denote the set $\{c \mid |c-a|\le |b|\}$, and $c=a\pm b$ to denote the condition
that $c$ is in the set $a\pm b$.
Let $[m]=\{1\ldots m\}$ for an integer $m \ge 1$. For $i \in [n]$, $A_{i*}$ denotes the $i$-th row of $A$ and for $j \in [d]$, $A_{*j}$ denotes the $j$-th column of $A$. We use bold symbols such as $\bA$, $\bS$ to emphasize that these objects are explicitly sampled from an appropriate distribution.

As mentioned, $\nnz(A)$ is the number of nonzero entries of $A$, and we assume $\nnz(A)\ge \rows$,
i.e., there are no rows composed entirely of zeros. 
We let $[A]_k$ denote the best rank-$k$ approximation to~$A$ in Frobenius norm and operator norm.
Further, for an $n \times d$ matrix $A$ and $S \subseteq [n]$, we use the notation $A_{S}$ to denote the restriction of the rows of $A$ to the subset indexed by $S$, and for $S \subseteq [d]$ we use the notation $A^{S}$ to denote the restriction of the columns of $A$ to the subset indexed by $S$.

Let $n^\omega$ be the time needed to multiply two $n\times n$ matrices. See \cite{demmel2007fast} and references therein for ways of computing other linear algebra primitives such as QR decomposition, SVD, and a matrix inverse, in  $O(n^{\omega})$ time. Given an $n \times d$ matrix $A$, $n \ge d$, we can orthogonalize its columns in time $O(nd^{\omega-1})$ as follows: first compute the product $\T{A}A$ in time $nd^{\omega-1}$, compute SVD of $\T{A}A$ in time $O(d^{\omega})$ to obtain $V,\Sigma$ such that $\T{A}A = V\Sigma^2\T{V}$, and then compute $AV\Sigma^{-1}$ in time $O(nd^{\omega-1})$ to obtain an orthonormal basis. 

For a matrix $A$, let $U$ be a matrix with orthonormal columns and $\text{colspan}(A) = \text{colspan}(U)$. The leverage score of the $i$-th row of $A$, $\ell_i^2$, is defined as $\opnorm{U_{i*}}^2$.

\begin{lemma}[Known constructions of sketching matrices]\label{lem embed}
For a given matrix $A\in\R^{\rows\times\colms}$ with $\rA=\rank(A)$, these constructions give $\eps$-embeddings
with failure probability $1-c$, for given constant~$c$. Here the sketching matrix $S$ is an $\eps$-embedding 
if with constant probability, $\opnorm{SAx} = (1\pm\eps)\opnorm{Ax}$ simultaneously for all $x\in\R^d$.
\begin{itemize}
\item
There is a sketching matrix $T\in\R^{m_T\times \rows}$ with sketching dimension $m_T = O(\eps^{-2}\rA^{1+\mu}\log \rA)$ such
that $TA$ can be computed in $O(\mu^{-1}\nnz(A)/\eps)$ time (see, e.g., \cite{cohen2016nearly}),
with $1/\mu\eps$ non-zero entries per column, in this form called here an \textsf{OSNAP},
and in earlier forms with $1$ non-zero per column called a \emph{CountSketch} \cite{CW13} matrix, 
or \emph{sparse embedding}. 
The sparsest version $\hat T\in\R^{m_{\hat T} \times \rows}$
has $m_{\hat T} = O(\eps^{-2}\rA^2)$, with
$\hat TA$ computable in $O(\nnz(A))$ time; $\hat T$ has one nonzero entry
per column. A less sparse version $\bar T$ of \textsf{OSNAP} has
$m_T = O(\eps^{-2}\rA \log (nd))$, $O(\log(nd)/\eps)$ entries per column,
and failure probability $1/\poly(\rows\colms)$.

\item
There is a sketching matrix $H\in\R^{m_H\times \rows}$ with $m_H = O(\eps^{-2}\rA\log(\rows\rA))$ such
that $HA$ can be computed in $O(\rows\colms\log\rows)$ time (see e.g. \cite{boutsidis2013improved}).
This is called an \textsf{SRHT} (Sampled Randomized Hadamard Transform) matrix.
The matrix $H= \hat H D$, where the rows of $\hat H$ are a random subset of the rows of a Hadamard matrix,
and $D$ is a diagonal matrix whose diagonal entries are $\pm 1$.
\item
If matrix $L\in \R^{m_L\times \rows}$ is chosen using leverage score sampling (see Theorem~\ref{thm:leverage-score-sampling}),
then there is $m_L = O(\eps^{-2}\rA\log \rA)$ so that $L$ is an $\eps$-embedding \cite{rudelson1999random, recht2011simpler}.
\item If matrix $G \in \R^{m_G \times n}$ with $m_G = O(\varepsilon^{-2}k)$ is an appropriately scaled matrix with i.i.d normal or Sub-Gaussian random variables, then $G$ is an $\varepsilon$ embedding.
\end{itemize}

These embeddings
can be composed, so that for example $S=H_ST_S$ is a ``two-stage'' $\eps$-embedding for $A$, where $T_S$ is an \textsf{OSNAP} matrix, and $H_S$ is an \textsf{SRHT},
so that $H_ST_SA$ can be computed in
$O(\eps^{-1}\mu\nnz(A) + \eps^{-2}n\rA^{1+1/\mu}\log^2(\rA/\eps))$ time, and the sketching dimension is
$m_{H_S}=O(\eps^{-2}\rA\log (\rA/\eps))$.
The space needed is $O(\rows +m_{H_S}\colms)$.
\end{lemma}

We also require the following notions of projection cost preserving sketches and affine embeddings.

\begin{Definition}[Projection Cost Preserving Sketch\cite{cohen2015dimensionality}]
\label{def PCP}
Given a matrix $A \in \mathbb{R}^{n \times d}$, $\eps>0$ and an integer $k\in [d]$, a sketch $SA \in\mathbb{R}^{s \times d} $ is a projection-cost preserving sketch of $A$ if for all rank-$k$ projection matrices $P$, 
\begin{equation*}
   (1-\eps)\frnorm{A (I - P)}^2\leq \frnorm{SA(I - P)}^2 \leq (1+\eps) \frnorm{A (I - P)}^2
\end{equation*}
\end{Definition}
We note that sometimes Projection Cost Preserving
Sketches allow an additive scalar in the definition,
see, e.g., \cite{mm20}. We do not need such an additive term here. 

\begin{Definition}[Affine Embeddings\cite{CW13}]
Given matrices $A,B$, let $X^* = \text{argmin}_X\frnorm{AX - B}$ and $\tilde{B} = AX^* - B$. A matrix $S$ is an affine embedding for $(A, B)$ if for all matrices $X$,
\begin{equation*}
    \frnorm{S(AX - B)}^2 - \frnorm{S\tilde{B}}^2 = (1 \pm \varepsilon)\frnorm{AX - B}^2 - \frnorm{\tilde{B}}^2.
\end{equation*}
\end{Definition}
Many subspace embedding distributions for the column space of $A$ satisfy the affine embedding property. Importantly, the number of rows in $S$ depends only on the rank of the matrix $A$ and has no dependence on number of columns in the matrix $B$. See \cite{CW13} for properties required of a distribution to be an affine embedding.

Throughout the paper, we use the following fact numerous times: for any matrices $A, B$,  and $C$,  we have $\frnorm{A - BC}^2 \ge \frnorm{A - AC^+C}^2$. This is just the Pythagorean theorem, which says that the best approximation of $A$ inside the rowspace of $C$ is obtained by projecting each of the rows of $A$ onto the rowspace of matrix $C$.

\section{Technical Overview}
The only known oblivious subspace embedding for a $k$ dimensional subspace with $o(k\log(k))$ rows is a dense matrix of $O(k)$ rows with independent Sub-Gaussian random variables. This embedding can be applied to a matrix $A$ in time $\Omega(\nnz(A) \cdot k)$. All other subspace embedding constructions that are faster to apply have at least $\Omega(k \log(k))$ rows. Obtaining a subspace embedding with few rows is important to speed up the further downstream tasks such as finding a maximal set of linearly independent rows of a matrix, computing approximate leverage scores, low rank approximation, etc. 

We analyze the properties required of a $k$-dimensional subspace $V \subseteq \R^d$, $d = \tilde{O}(k)$, such that a sparse random sign matrix with $o(k\log(k))$ rows can be a subspace embedding for $V$. The advantage of the sparsity is that the embedding can be applied to a vector quickly. Suppose every unit vector in the subspace $V$ has at least a constant $c$ fraction of coordinates that have a magnitude of at least $\tilde{\Omega}(1/\sqrt{k})$. Let $x$ be an arbitrary unit vector in the subspace $V$. Now consider a random matrix $\bG$ where each entry is either $0$ with  probability $1-p$ and $\pm 1$ with probability $p/2$ each. For $p = \Theta(1/d)$, as at least a constant $c$ fraction of the coordinates of the vector $x$ have a magnitude $\tilde{\Omega}(1/\sqrt{k})$, each row of the matrix $\bG$ has $\Omega(1)$ probability of hitting one of the large coordinates of the vector $x$. Conditioned on a row $\bG_{i*}$ hitting one of the large coordinates of $x$, we have $|\bG_{i*}x| \ge \tilde{\Omega}(1/\sqrt{k})$ with probability $\ge 1/2$ by using the random signs. Thus, with at least a constant probability, for a row $\bG_{i*}$, $|{\bG_{i*}x}|^2 \ge \tilde{\Omega}(1/k)$. If the matrix $\bG$ has $\Omega(k)$ rows, using the Chernoff bound, we have that with very high probability, $\opnorm{\bG x}^2 \ge \tilde{\Omega}(1)$, which suffices to union bound over a suitable net of unit vectors in a $k$-dimensional subspace. On the other hand, showing that $\opnorm{\bG}$ is small and that it does not increase the norm of any unit vector  by a lot is much easier. For  the probability $p$ that we consider, each row and column of the matrix $\bG$ only has $O(1)$ nonzero entries with high probability. As all the nonzero entries are at either  $\pm 1$, we can bound the operator norm $\opnorm{\bG}$ by $O(1)$. This implies that for any unit vector $x$, $\opnorm{\bG x}^2 \le O(1)$.

The above argument shows that if a subspace has the property that every unit vector in the subspace has a \emph{large} number of \emph{large} coordinates, then a random sparse sign matrix is a subspace embedding with small distortion for that subspace. We call subspaces having this property \emph{flat}. But of course, the column space of the matrix to which we want to apply the embedding may not have this property. Let $V_1 \subseteq \R^n$ be the column space of the given matrix $A$. If we can find a linear map $\calF$ that maps vectors in the subspace $V_1$ to a \emph{flat} subspace $V_2$ and if $\calF$ preserves the Euclidean norms of the vectors, then we have that $\opnorm{\bG \calF x} \approx \opnorm{\calF x} \approx \opnorm{x}$ for all vectors $x \in V_1$. As we show later, by paying some cost in running time, we can assume that $n = O(k\log(k))$ by first applying a series of suitable \textsf{OSNAP} embeddings. To obtain such a mapping $\calF$, we use the $\ell_2 \rightarrow \ell_1$ embedding $F$ of \cite{indyk2007uncertainty}.
We show that recursively applying the linear map $F$ gives a linear map $\calF : n \rightarrow n^{1+o(1)}$ with the property that for all unit vectors $x$, $\opnorm{\calF x} \approx 1$ and $\|{\calF x}\|_1 \ge \tilde{\Omega}(\sqrt{n})$. This property immediately shows that the vector $\calF x$ must have a large number of large coordinates and therefore that the subspace range($\calF$) is \emph{flat}. We only obtain that a $1/n^{o(1)}$ fraction of the coordinates are large but it is sufficient for our purposes. We also show that the sequence of \textsf{OSNAP}, the mapping of \cite{indyk2007uncertainty} which we call Indyk, and the sparse random sign embeddings can be applied to a matrix $A \in \R^{n \times k}$ in time $O(\gamma^{-1}\nnz(A) + k^{2+\gamma+o(1)})$ for any constant $\gamma > 0$. 

\paragraph{$1+\varepsilon$ Subspace Embeddings.} We use our $\exp(\poly(\log\log k))$ distortion subspace embedding construction to obtain $1+\varepsilon$ \emph{non-oblivious} subspace embeddings using approximate leverage scores obtained by using a preconditioner. Let $A \in \R^{n \times k}$. Earlier algorithms to compute approximate leverage scores can be described as follows : (i) Compute $\bS A$ where $\bS$ is a subspace embedding for the column space of $A$, (ii) Compute an orthonormal matrix $Q$ and matrix $R^{-1}$ such that $\bS A = QR^{-1}$, and (iii) Compute the approximate leverage scores $\tilde{\ell}_i^2 = \opnorm{A_{i}R}^2$. 

Thus, to make computing approximate leverage scores faster, we need a subspace embedding $\bS$ that can be quickly applied to matrix $A$ to make step (i) faster while also having a fewer number of rows to make the computation of the QR-decomposition in step (ii) faster. As discussed, our subspace embedding construction $\bS$ has both of these desired properties. In step (iii), instead of computing $\opnorm{A_{i*}R}^2$ exactly, a Gaussian matrix $\bG$ with $O(\log(n))$ columns is used so that for all the rows $i \in [n]$, $\opnorm{A_{i*}R\bG}^2 \approx \opnorm{A_{i*}R}^2$, which is a standard idea \cite{DMMW12}. However, computing the matrix $AR\bG$ takes $\Omega(\nnz(A)\log(n))$ time. We consider using a Gaussian matrix with only $O(1/\gamma)$ columns for an absolute constant $\gamma > 0$, which is also a standard idea in this area. Consider an arbitrary vector $v$ and let $\bg$ be a vector of i.i.d. normal random variables. Then we have the probability that $|\langle v, \bg\rangle| \le \opnorm{v}/n^{\gamma}$ is at most $1/n^\gamma$. If $\bg_1, \ldots, \bg_t$ are independent Gaussian vectors for $t = O(1/\gamma)$, then at least one of the values $|\langle v, \bg_i\rangle|$ is at least $\opnorm{v}/n^\gamma$ with probability $\ge 1 - 1/n^2$. If $\bG$ is a matrix with $\bg_j$ as its columns, we therefore have that $\opnorm{A_{i*}R\bG}^2 \ge \opnorm{A_{i*}R}^2/n^{2\gamma}$ for all $i$. We also argue that $\opnorm{A_{i*}R\bG}^2 = O(\opnorm{A_{i*}R}^2\log(n))$ for all $i \in [n]$. Now the matrix $AR\bG$ and the approximations $\opnorm{A_{i*}R\bG}^2$ can be computed in time $O(\gamma^{-1}(\nnz(A)+k^2))$.
Therefore we can obtain over-estimates to the leverage scores. Using over-estimates to the leverage score sampling probabilities, we first sample rows and then compute accurate leverage scores only for the rows that are sampled.
Then we employ a rejection step, in which we reject rows randomly based on the probabilities computed using accurate leverage scores, and finally we show that we obtain a sample from the leverage score sampling distribution. As we compute accurate leverage scores only for the rows that are sampled in the first stage, we do not incur the $O(\nnz(A)\log(n))$ factor. We then compose our leverage score embedding with an OSNAP embedding to obtain a $1+\varepsilon$ embedding with $O(k\log(k)/\varepsilon^2)$ rows, which is faster than previous constructions.

\paragraph{Computing Linearly Independent Rows.} We give an algorithm to compute a maximal set of linearly independent rows of a matrix $A \in \R^{n \times d}$ of rank $k$ in time $O(\nnz(A) + k^{\omega}\poly(\log\log(n)))$. Using the rank-preserving sketches of \cite{CKL13}, we can assume without loss of generality that $d = ck$ for a constant $c$. The crucial idea here is that a leverage score sample of the matrix $A$, with high probability, must contain a set of $k$ linearly independent rows. Therefore, directly applying the above leverage score sampling algorithm for constant $\varepsilon$ gives, in time $O(\gamma^{-1}\nnz(A) + n^{\gamma}k^{2+o(1)} + k^{\omega}\poly(\log\log(n)))$, for any constant $\gamma$, a set of $O(k\exp(\poly(\log\log k)))$ rows of the matrix $A$ that must contain a set of $k$ linearly independent rows. To obtain a running time that does not depend on $\gamma$, we show that instead of running leverage score sampling on the matrix $A$, we can apply reductions as in \cite{CKL13} to reduce the problem to computing linearly independent rows of a sub-matrix $A'$ with $\nnz(A') \le \nnz(A)/\poly(\log(n))$ and with $n/\poly(\log(n))$ rows. This reduction can be performed in time $O(\nnz(A) + k^{\omega}\log\log(n))$. After this reduction, we perform leverage score sampling for the matrix $A'$ as described above with constant $\varepsilon$ and $\gamma = O(1/\log(n))$ to obtain a matrix $\bS_{\lev}$ that selects and scales $O(k\exp(\poly(\log\log k)))$ rows randomly according to the leverage score distribution such that for all $x$, $\opnorm{\bS_{\lev}A'x} = (1 \pm 1/2)\opnorm{A'x}$. In particular, the guarantee implies that $\text{rowspace}(\bS_{\lev}A') = \text{rowspace}(A')$. Therefore there are $k$ linearly independent rows among the $O(k\exp(\poly(\log\log k)))$ rows sampled by $\bS_{\lev}$. Now we can again apply the recursive row reduction procedure mentioned above to the matrix $\bS_{\lev}A'$, to finally obtain, in time $O(k^{2+o(1)} + k^{\omega}\poly(\log\log k))$, a set of $O(k)$ rows that, with high probability, contain a set of $k$ linearly independent rows. These rows can now be identified in time $O(k^{\omega})$. Thus,  we obtain that in time $O(\nnz(A) + k^{\omega}\poly(\log\log n) + k^{2+o(1)})$, we can compute a set of $k$ linearly independent rows of a rank $k$ matrix~$A$. As discussed above, the subspace embedding having $k\poly(\log\log k)$ rows turns out to be crucial to obtain a running time that depends on $k^{\omega}\poly(\log\log n)$ instead of the $k^{\omega}\log(n)$ dependence of earlier algorithms.

\paragraph{Low Rank Approximation.} Finally, we give an algorithm to compute a $(1+\varepsilon)$-approximate rank-$k$ approximation to an arbitrary matrix $A$. We note that we do not need to utilize our subspace embedding construction in this algorithm, though we include it as it is also a fundamental problem in linear algebra for which we remove the last logarithmic factor. We compute a low rank approximation in two stages: (i) we first find a rank $k$ orthonormal matrix $V$ whose columns span a $1+\varepsilon$ approximation. (ii) we then find a right factor $\tilde{X}$ such that $V \cdot \tilde{X}$ is a $(1+\varepsilon)$ rank-$k$ approximation. We obtain the left factor $V$ by using projection-cost preserving sketches and subspace embeddings along with the CUR decomposition algorithm from \cite{boutsidis2017optimal}, to first obtain an $O(k)$-dimensional subspace that spans an $O(1)$-approximate rank-$k$ low rank approximation. We then perform the residual sampling algorithm of \cite{deshpande2006matrix} to obtain a set of $O(k/\varepsilon)$ columns of the matrix $A$,  which along with the $O(k)$ dimensional subspace we already found, span a $(1+\varepsilon)$-approximation. We then use affine embeddings to compute a left factor $V$ that spans a $(1+\varepsilon)$-approximation. 

After finding a left factor $V$, the matrix $\T{V}A$ is the optimal right factor but it takes $\Omega(\nnz(A) \cdot k)$ time to compute this matrix. We then run the CUR decomposition algorithm of \citet{boutsidis2017optimal} using the matrix $V$ we found to obtain a right factor $\tilde{X}$ such that $\frnorm{V \cdot \tilde{X} - A} \le (1+\varepsilon)\frnorm{A - [A]_k}$. 

\section{Flattening the vectors}
In this section, we argue that there is a linear mapping $\calF : \R^{n} \rightarrow \R^{n^{1+o(1)}}$ such that for any unit vector $x \in \R^{n}$, the set
\begin{equation*}
	\text{Large}(\calF x) := \{i \in [n^{1+o(1)}]\,|\, |(\calF x)_i| \ge \frac{1}{\sqrt{n} \cdot \epll{n}}\}
\end{equation*}
has size $|\text{Large}(\calF x)| = \Omega(n)$.

We show that an explicit $\ell_2 \rightarrow \ell_1$ linear embedding construction of Indyk \cite{indyk2007uncertainty} can be used to obtain such a mapping $\calF$. First we define $(\varepsilon, l)$ extractors as follows.

\begin{Definition}[$(\varepsilon, l)$ extractors]
	A bipartite graph $G = (A, B, E)$, $A = [a]$ and $B = [b]$, with each left node having degree $d$ is an $(\varepsilon, l)$ extractor if it has the following property. Let $\calP$ be \emph{any} distribution over the set $A$ such that for all $i \in [a]$, $\Pr_{\calP}[i] \le 1/l$. Consider the distribution over $B$ generated by the following process:
	\begin{enumerate}
		\item Sample $i \in A$ from distribution $\calP$
		\item Sample $t \in [d]$ uniformly at random and set $j = \Gamma_{G}(i)_t$. Here $\Gamma_{G}(i)$ is the ordered set of neighbors of $i$ in the graph $G$ and $\Gamma_G(i)_t$ is the $t$-th neighbor in the ordered set.
	\end{enumerate}
	Let $G(\calP)$ be the distribution of the element $j$ sampled by the above process and let $\calI$ be the uniform distribution over the set $B$. The graph $G$ is an $(\varepsilon, l) $ extractor if $\sum_{j \in B}|\Pr_{G(\calP)}[j] - 1/b| \le \varepsilon$. We stress that this property must hold for every distribution $\calP$ with $\Pr_{\calP}[i] \le 1/l$ for all $i$.
\end{Definition}

See \cite{indyk2007uncertainty} and references therein for explicit constructions of extractors. Indyk uses the following extractor: Fix a $\delta = \Omega(1/\sqrt{n})$ and let $L = O(1/\delta^2)$ and $s = \sqrt{n}$. Let $G$ be an $(\varepsilon, l)$ extractor with $A = [Ln]$, $B = [b]$ for $b = n^{1/2-\kappa}$, $\kappa > 0$, $l = (1-\delta)^2s/L$, left degree $d = (\log a)^{O(1)} = (\log Ln)^{O(1)}$ and right degree $\Delta = O(nLd/b)$.

In the following it will be helpful to have an abbreviation.

\begin{Definition}
Let $\epll{n}$ denote the class of functions in $\exp(\poly(\log\log(n)))$ as integer $n\rightarrow\infty$.
\end{Definition}

\begin{theorem}[Theorem~1.1 of \cite{indyk2007uncertainty}]
	For any $\zeta, \kappa > 0$, there is an explicit linear mapping $F : \R^{n} \rightarrow \R^{m}$, $m = O(nLd) = n\log^{O(1)}(n)/\zeta^{O(1)}$ and a partitioning of the coordinate set $[m]$ into sets $B_1, \ldots, B_b$, for $b = n^{1/2-\kappa}$, each of size at most $\Delta = n^{1/2+\kappa} \epll{n}/\zeta^{O(1)}$, such that for any $x \in \R^n$, $\opnorm{x}=1$,
	\begin{equation*}
		(1-O(\zeta))\sqrt{Ldb} \le \sum_{j=1}^b \opnorm{(Fx)_{B_j}} \le \sqrt{Ldb}.
	\end{equation*}
Without loss of generality, we can assume that all the partitions $B_j$ have the same size $\Delta$ by appending $0$-valued coordinates and so we have $m = n \cdot \epll{n}/\zeta^{O(1)}$.
\end{theorem}

We now prove the following lemma which essentially shows that an application of Indyk's embedding  to a unit vector shrinks the Euclidean norm by a lot, while keeping the $\ell_1$ norm $\Omega(1)$.
\begin{lemma}
	Let $n$ be an arbitrary integer and $0 < \zeta, \kappa < c$ for a small enough constant $c$. There is an explicit linear mapping $F: \R^n \rightarrow \R^m$ for $m = n\cdot\epll{n}/\zeta^{O(1)}$ and a partitioning of $[m]$ into equal sized sets $B_1,\ldots, B_{b}$ where
$
		b = n^{1/2-\kappa}
$
	and each set $B_j$ satisfies
$
		|B_j| = \Delta = n^{1/2 + \kappa}\epll{n}/\zeta^{O(1)},
$
	such that for any $x \in \R^n$, we have
	\begin{equation*}
		(1-O(\zeta))\opnorm{x} \le \sum_{j=1}^{b} \opnorm{(Fx)_{B_j}} \le \opnorm{x}
	\end{equation*}
	and
	\begin{equation*}
		\frac{1}{b}(1-O(\zeta))\opnorm{x}^2 \le \opnorm{Fx}^2 = \sum_{j=1}^{b}\opnorm{(Fx)_{B_j}}^2 \le \frac{1}{b}\opnorm{x}^2.
	\end{equation*}
	\label{lma:indyk-embedding-lemma}
\end{lemma}
\begin{proof}
In the proof of the above theorem, Indyk uses the  $(\varepsilon, l)$ construction specified above with $\delta = \zeta$ and $\varepsilon = \zeta^2$. Indyk also defines $(Fx)_{B_j} := (Dx)_{\Gamma_G(j)}$ for $j \in [b]$, where $D$ is a concatenation of certain $L$ orthonormal matrices and $\Gamma_G(j) \subseteq A$ is the set of neighbors of $j \in B$ in the graph $G$. For any unit vector $x$, we have $\opnorm{Dx}^2 = L$ and as the left degree of $G$ is exactly equal to $d$, we have $\opnorm{Fx}^2 = \sum_{j}\opnorm{(Fx)_{B_j}}^2 = \sum_j \opnorm{(Dx)_{\Gamma_G(j)}}^2 = d\opnorm{Dx}^2 = Ld$.

Let $y = Dx$ and let $S$ be the set of the $s$ largest magnitude entries of $y$. Define $z = y_{[a]-S}$ where $z$ is obtained by zeroing out the coordinates of the set $S$. \citet{indyk2007uncertainty} showed that
\begin{equation*}
	\sum_{j=1}^b \left|\frac{1}{\rho^2d}\opnorm{z_{\Gamma_{G}(j)}}^2 - 1/b\right| \le \epsilon
\end{equation*}
where $\rho \ge \sqrt{L}(1-\delta)$. The inequality implies that
$
	\sum_j \opnorm{z_{\Gamma_{G}(j)}}^2 \ge \rho^2d(1-\epsilon) \ge Ld (1-\delta)^2(1-\epsilon).
$
As $\opnorm{y_{\Gamma_G(j)}} \ge \opnorm{z_{\Gamma_G(j)}}$, we get
$
	\sum_j \opnorm{y_{\Gamma_G(j)}}^2 \ge Ld(1-\delta)^2(1-\epsilon)
$
and plugging in $\delta = \zeta$ and $\epsilon = \zeta^2$,  we obtain
$
	Ld(1-O(\zeta)) \le \sum_{j=1}^b \opnorm{(Fx)_{B_j}}^2 \le Ld.
$
Hence, the matrix $F/\sqrt{Ldb}$ satisfies that for any vector $x$,
\begin{equation*}
	\frac{1}{b}(1-O(\zeta))\opnorm{x}^2 \le \sum_{j=1}^b \opnorm{(\frac{F}{\sqrt{Ldb}}x)_{B_j}}^2 \le \frac{1}{b}\opnorm{x}^2.
\end{equation*}
From the above theorem, we already have
\begin{equation*}
	(1-O(\zeta))\opnorm{x} \le \sum_{j=1}^b \opnorm{(\frac{F}{\sqrt{Ldb}}x)_{B_j}} \le \opnorm{x}.
\end{equation*}
Therefore, scaling the matrix $F$ gives the proof.
\end{proof}

We apply the above lemma recursively to each of the partitions $B_j$ for $\Theta(\log\log(n))$ levels to obtain the following theorem. 

\begin{theorem}
	Given any $n$, there is an explicit map $\calF : \R^n \rightarrow \R^m$ with $m = n\cdot \epll{n}$ such that for all unit vectors $x \in \R^n$, we have
	\begin{equation*}
		\|\calF x\|_1 \ge \frac{\sqrt{n}}{4}	
	\end{equation*}
	and
	\begin{equation*}
		\frac{1}{2} \le \opnorm{\calF x}^2 \le 1.
	\end{equation*}
	Further, given any vector $x$, the vector $\calF x$ can be computed in $n^{1+o(1)}$ time.
	\label{thm:indyk-multi-level}
\end{theorem}
\begin{proof}
Let $N = \Theta(\log\log(n))$ and $\zeta = \Theta(1/\log\log(n))$. Let $B_1,\ldots, B_{b_1}$ be the partitions of the coordinates of the range of $F$ from the Lemma~\ref{lma:indyk-embedding-lemma}. We recursively apply the lemma for each of the partitions for $N$ levels to obtain $\calF : \R^{n} \rightarrow \R^{m}$ for $m = n\cdot \epll{n}$. Define $n_0 = n$ and let $n_i$ be the number of entries in each of the $i$-th level partitions. Also, let $b_0 = 1$ and $b_i$ be the number of partitions an $(i-1)$-th level partition is mapped into. From Lemma~\ref{lma:indyk-embedding-lemma}, we have
\begin{equation*}
	b_i = n_{i-1}^{1/2-\kappa}
\end{equation*}
and
\begin{equation*}
	n_i = n_{i-1}^{1/2+\kappa}\epll{n_{i-1}}/\zeta^{O(1)}.
\end{equation*}
The following lemma lower bounds the number of partitions in the $N$-th level.
\begin{lemma}
	The total number of partitions in the $N$-th level is given by $B = b_0 \cdot b_1 \cdots b_N$ and 
	\begin{equation*}
		B \ge n/2.
	\end{equation*}
\end{lemma}
\begin{proof}
	We have $B = b_1 \cdots b_N = (n_0 \cdots n_{N-1})^{1/2 - \kappa}$. As $n_i \ge n^{(1/2 + \kappa)^i}$, we have that $n_0 \cdots n_{N-1} \ge n^{\sum_{i=0}^{N-1}(1/2+\kappa)^{i}}$. Now,
$
		\sum_{i=0}^{N-1}(1/2+\kappa)^i = {(1 - (1/2+\kappa)^N)}/{(1/2 - \kappa)}
$
	which implies $B \ge n^{1 - (1/2+\kappa)^{N}}$. For $N = \Theta(\log\log(n))$, $(1/2 + \kappa)^N \le 1/\poly(\log(n))$ and $B \ge n/2$.
\end{proof}
This lemma implies that the $N$-th level has the partitions $\calB_1, \ldots, \calB_B$ of $[m]$  with $B \ge n/2$ and $|\calB_j| = \epll{n}$ such that for any unit vector $x$, 
\begin{equation*}
	\frac{1}{2}\opnorm{x} \le (1-O(\zeta))^{N}\opnorm{x} \le \sum_{j=1}^B\opnorm{(\calF x)_{\calB_j}} \le \opnorm{x}
\end{equation*}
and
\begin{equation*}
	\frac{1}{2B}\opnorm{x}^2 \le \frac{(1-O(\zeta))^N}{B}\opnorm{x}^2 \le \sum_{j=1}^B\opnorm{(\calF x)_{\calB_j}}^2 \le \frac{1}{B}\opnorm{x}^2.
\end{equation*}
Finally, for a unit vector $x$,
\begin{equation*}
\frac{1}{2} =	\frac{1}{2}\opnorm{x} \le \sum_{j=1}^B \opnorm{(\calF x)_{\calB_j}} \le \sum_{j=1}^B \|{(\calF x)_{\calB_j}}\|_1 = \|\calF x\|_1
\end{equation*}
and
\begin{equation*}
	\frac{1}{2B} = \frac{1}{2B}\opnorm{x}^2 \le \opnorm{\calF x}^2 = \sum_{j=1}^B \opnorm{(\calF x)_{\calB_j}}^2 \le \frac{1}{B}\opnorm{x}^2 = \frac{1}{B}.
\end{equation*}
By scaling the map $\calF$ by $\sqrt{B}$, we complete the proof.
\end{proof}

We now have the following corollary.
\begin{corollary}
	Given any unit vector $x$, at least $\Theta(n)$ coordinates of the vector $\calF x \in \R^m$ have an absolute value of at least $\eta = {1}/{(\sqrt{n}\cdot \epll{n})}$.
	\label{cor:indyk-multi-level}
\end{corollary}
\begin{proof}
	Let $m'$ be the number of coordinates of $\calF x$ with an absolute value of at least $\eta$. Let $T \subseteq [m]$  be the set of indices of those coordinates. Then
\begin{align*}
	\frac{1}{4}\sqrt{n} \le \|\calF x\|_1 &= \sum_{i \notin T}|(\calF x)_i| + \sum_{i \in T}|(\calF x)_i|\\
	&\le \frac{m}{\sqrt{n}\cdot \epll{n}} + \sqrt{\sum_{i \in T} (\calF x)_i^2}\sqrt{|T|}\\
	&\le \frac{n \cdot \epll{n}}{\sqrt{n}\cdot \epll{n}} + \sqrt{m'}.
\end{align*}
Here we use the Cauchy-Schwarz inequality and the fact that $\opnorm{\calF x}^2 \le 1$.
For appropriate $\eta$ chosen based on $m$, the above inequality implies that
\begin{equation*}
	\sqrt{m'} \ge \sqrt{n}/8 \implies m' \ge n/64
\end{equation*}
which shows that an $\Omega(n)$ fraction of the coordinates of $\calF x$ have an absolute value of at least $\eta$.
\end{proof}
Thus, applying Lemma~\ref{lma:indyk-embedding-lemma} for $N = \Theta(\log\log(n))$ levels gives an $n$ dimensional subspace of $\R^m$ for $m = n\cdot\epll{n}$ such that for every unit vector $x$, the vector $\calF x$ has a \emph{large} number of \emph{large} coordinates.
\section{Fast Subspace Embeddings}
\begin{algorithm2e}
\caption{\textsc{FastEmbedding}}
\KwIn{$A \in \R^{n \times k}, \gamma > 0$}
\KwOut{A subspace embedding $\bS A$ with $O(k\cdot\epll{k})$ rows}
\DontPrintSemicolon
$\bS_1 \gets $ \textsf{OSNAP}($A$, $\gamma$) with $O(k^{1+\gamma+o(1)})$ rows\;
$\bS_2 \gets $ \textsf{OSNAP}($\bS_1 A$, $O(1/\log(n))$) with $O(k\log(k))$ rows\;
$\calF \gets $ Indyk Embedding for $\R^{O(k\log(k))}$ for $\Theta(\log\log(k))$ levels with $r = k\cdot\epll{k}$ rows\;
$m \gets k \cdot \poly(\log\log k)$, $p \gets \epll{k}/r$\;
$\bG \gets $ $m \times r$ random matrix where each entry is independently $0$ with probability $1-p$, and $\pm 1$ with probability $p/2$ each\;
$\bS A \gets \kappa \cdot \bG \cdot \calF \cdot \bS_2 \cdot \bS_1 A$ where $\kappa$ is an appropriate scaling factor\;
\Return{$\bS A$}
\end{algorithm2e}
Let $A$ be an arbitrary $n \times k$ matrix with $\nnz(A)$ nonzero entries. We design a random matrix $\bS$ with $k \cdot \poly(\log\log(k))$ rows such that with probability $\ge 9/10$, for all vectors $x$,
\begin{equation*}
\opnorm{x}	\le \opnorm{\bS Ax} \le \epll{k} \opnorm{x}.
\end{equation*}
The matrix $\bS A$ can be computed in time $\nnz(A) + k^{2.1+o(1)}$. The matrix $\bS$ is constructed as a composition of various oblivious subspace embeddings. 

We first apply \textsf{OSNAP} $\bS_1$ with $\mu = 0.1$ to obtain an $O(k^{1.1}\log(k)) \times k$ matrix $\bS_1A$ in time $O(\nnz(A))$. Now, $\nnz(\bS_1A) = O(k^{2.1}\log(k))$. Therefore, we can apply \textsf{OSNAP} $\bS_2$ with $\mu = 1/\log(k)$, to obtain an $O(k\log k)\times k$ matrix $\bS_2\bS_1A$ in time $O(\nnz(\bS_1 A) \cdot 1/\mu) = O(k^{2.1}\log^2(k))$. We also have with probability $\ge 98/100$ that
\begin{equation*}
	\opnorm{\bS_2 \bS_1Ax} \in (1 \pm 3/10)\opnorm{Ax}
\end{equation*}
for all vectors $x \in \R^k$. We then use the flattening transform $\calF$ to obtain a constant subspace embedding for the matrix $\bS_2 \cdot \bS_1 \cdot A$ which also has the property that every unit vector in the column space of the matrix $\calF \cdot \bS_2 \cdot \bS_1 \cdot A$ has a large number of large entries.
\begin{theorem}[Indyk Embedding, Theorem~\ref{thm:indyk-multi-level} and Corollary~\ref{cor:indyk-multi-level}] Given any $n$, there is an explicit linear map/matrix $\calF \in \R^{m \times n}$ with $m = n \cdot \epll{n}$ such that for any vector $x \in \R^n$,
\begin{equation*}
	\frac{1}{2}\opnorm{x} \le \opnorm{\calF x} \le \opnorm{x}
\end{equation*}
and for any unit vector $x$, at least $\Theta(n)$ coordinates of the vector $\calF x$ have an absolute value of at least $1/(\sqrt{n} \cdot \epll{n})$. Given a vector $x \in \R^n$, the explicit map $\calF x$ can be computed in time $n^{1+o(1)}$.	
\label{thm:final-indyk-version}
\end{theorem}
Combining $\calF$, $\bS_2, \bS_1$, we obtain that with probability $\ge 98/100$, for all vectors $x$,
\begin{equation*}
\frac{1}{4}\opnorm{Ax} \le \opnorm{\calF \cdot \bS_2 \cdot \bS_1 \cdot Ax} \le \frac{3}{2}\opnorm{Ax}.
\end{equation*}
The matrix $\calF \cdot \bS_2 \cdot \bS_1 \cdot  A$ can be computed in time $\nnz(A) + k^{2.1+o(1)}$. As the matrix $\bS_{2} \cdot \bS_1 \cdot A$ has $O(k \cdot \log(k))$ rows, the matrix $\calF$ has $O(k \log(k) \cdot \epll{k}) = k\cdot \epll{k}$ rows and we also obtain that for any unit vector $x$ in the column space of $\calF \cdot S_2 \cdot S_1 \cdot A$, at least $\Theta(k\log k)$ coordinates have an absolute value of at least $1/(\sqrt{k\log k}\,\epll{k}) = 1/(\sqrt{k}\,\epll{k})$. The following theorem shows that a sparse sign matrix is a subspace embedding for a subspace with every unit vector in the subspace having a large number of large entries.
\begin{theorem}
\label{thm:sparse-embedding}
	Let $A \in \R^{m \times k}$, with $m = k\cdot \epll{k}$, be a matrix such that for all unit vectors $x \in \text{colspan}(A)$, the set 
	\begin{equation*}
		\text{Large}(x) := \left\{i \in [m]\,\mid \, |x_i| \ge \eta = \frac{1}{\sqrt{k} \cdot \epll{k}}\right\}
	\end{equation*}
	satisfies $|\text{Large}(x)| \ge Ck$ for some constant $C$. There is a distribution $\calG$  over matrices with $M = k \cdot \poly(\log\log(k))$ rows such that for $\bG \sim \calG$, with probability $\ge 9/10$, for all vectors $x \in \R^k$,
	\begin{equation*}
		\opnorm{Ax} \le \opnorm{\bG Ax} \le \epll{k}\opnorm{Ax}.
	\end{equation*}
	With probability $\ge 9/10$, the matrix $\bG A$ can be computed in time $k^2\cdot\epll{k}$.
\end{theorem}
\begin{proof}
	Define the $M \times m$ random matrix $\bG$ as follows:
	\begin{equation*}
		\bG_{ij} = \begin{cases}
					+1 & \text{with probability $p/2$}\\
					-1 & \text{with probability $p/2$}\\
					0 & \text{with probability $1-p$}
			\end{cases}
	\end{equation*}
	for some values of $M \le m$ and $p$ to be chosen later. The random variables $\bG_{ij}$ are mutually independent. Let $\bX_i$ be the number of nonzero entries in the $i$-th row of $\bG$ and let $\bY_j$ be the number of nonzero entries in the $i$-th column of $\bG$. By the Chernoff bound, for $\delta > 1$, 
	\begin{equation*}
		\Pr[\bX_i \ge (1+\delta) \cdot mp] \le \exp(-\delta mp/4)
\quad \text{and} \quad
		\Pr[\bY_j \ge (1+\delta) \cdot Mp] \le \exp(-\delta Mp/4).
	\end{equation*}
	Let $p$ be such that $p|\text{Large}(x)| \ge 10$ for all $x$. As $|\text{Large}(x)| \ge Ck$, there is a value of $p$ for which $pm \le \epll{k}$. By a union bound, we obtain that with probability $\ge 99/100$, for all $i$ and $j$, $\bX_i \le \epll{k}$ and $\bY_j \le  \epll{k}$. Thus, with probability $\ge 99/100$
	\begin{equation*}
	\max_{i} \sum_j |\bG_{ij}| = \max_i \bX_i \le \epll{k}
\, \text{and} \,
		\max_{j} \sum_i |\bG_{ij}| = \max_j \bY_j \le \epll{k}.
\end{equation*}
We now have that $\opnorm{\bG} \le \sqrt{(\max_i\sum_j|\bG_{ij}|)(\max_j \sum_{i} |\bG_{ij}|)} \le \epll{k}$, which implies that for any vector $y$,
\begin{equation*}
	\opnorm{\bG \cdot Ay} \le \epll{k}\opnorm{Ay}.
\end{equation*}
Let the event that $\opnorm{\bG} \le \epll{k}$ be $\calE$. 

We now show a contraction lower bound. Let $x$ be an arbitrary unit vector in the column space of the matrix $A$. We say a row $\bG_{i*}$ is \emph{good} if $\bG_{ij}$ is nonzero for some $j \in \text{Large}(x)$. We say $\bG_{i*}$ is \emph{bad} if it is not \emph{good}. We have
\begin{equation*}
	\Pr[\text{$\bG_{i*}$ is \emph{bad}}] = (1-p)^{|\text{Large}(x)|} \le \exp(-p|\text{Large}(x)|) \le \exp(-10) \le 1/100.
\end{equation*}
Thus, $\Pr[\text{$\bG_{i*}$ is \emph{good}}] \ge 99/100$. 

We say a row $\bG_{i*}$ is \emph{large} if $|G_{i*}x| \ge \eta$. Condition on the event that $\bG_{i*}$ is \emph{good}. Let $j \in \text{Large}(x) \cap \nnz(\bG_{i*}) \ne \emptyset$. Now, $\bG_{i*}x = \sum_{j' \in \nnz(\bG_{i*})-j}\bG_{ij'}x_{j'} + \bG_{ij}x_{j}$. As entries of the matrix $\bG$ are mutually independent, with probability $1/2$, $\bG_{ij}x_j$ has the same sign as $\sum_{j' \in \nnz(\bG_{i*})-j}\bG_{ij}x_j$,  which implies that with probability $\ge 1/2$, $|\bG_{i*}x| \ge |x_j| \ge \eta$. Thus,
\begin{equation*}
	\Pr[\text{$\bG_{i*}$ is \emph{large}}\,|\, \text{$\bG_{i*}$ is \emph{good}}] \ge 1/2
\end{equation*}
which implies that
\begin{equation*}
	\Pr[|\bG_{i*}x| \ge \eta] = \Pr[\text{$\bG_{i*}$ is \emph{large}}] \ge (1/2) \cdot (99/100) \ge 1/4.
\end{equation*}
Let $l$ denote the number of \emph{large} rows. As rows of the matrix $\bG_{i*}$ are independent, \emph{largeness} of rows is mutually independent. Thus, by the Chernoff bound, 
\begin{equation*}
	\Pr[l \le (1/2) \cdot M \cdot (1/4)] \le \exp(-M/32).
\end{equation*}
We now condition on the event $\calE$. We have
\begin{equation*}
	\Pr[l \le M/8\,|\, \calE] \le \frac{\Pr[l \le M/8]}{\Pr[\calE]} \le 2\exp(-M/32).
\end{equation*}
Therefore, conditioned on the event $\calE$, with probability $\ge 1 - 2\exp(-M/32)$, we have $l \ge M/8$ which implies that
\begin{equation*}
	\opnorm{\bG x}^2 \ge \sum_{\text{\emph{large}}\, i}|\bG_{i*}x|^2 \ge l\eta^2 \ge \frac{l}{k\,\epll{k}} \ge \frac{M}{8k\,\epll{k}}.
\end{equation*}
In what follows, we condition on the event $\calE$. For $M = k \cdot \poly(\log\log(k))$, we obtain that for a unit vector $x$, with probability $\ge 1-\exp(-k\poly(\log\log(k)))$,
\begin{equation*}
	\opnorm{\bG x}^2 \ge \frac{\poly(\log\log(k))}{\epll{k}}.
\end{equation*}
By suitably scaling $\bG$, we obtain that for all vectors $x$,
\begin{equation*}
\opnorm{\bG x} \le \epll{k}\opnorm{x}
\end{equation*}
and for any unit vector $x$, with probability $\ge 1 - \exp(-k \cdot \poly(\log\log(k)))$, 
\begin{equation*}
	\opnorm{\bG x} \ge 2.
\end{equation*}
The column space of the matrix $A$ has dimension at most $k$. Let $\calN$ be a net of the unit vectors in the column space of $A$ such that for any $y \in \text{colspace}(A)$, $\opnorm{y} = 1$, there is an $x_y \in \calN$, $\opnorm{x_y}=1$ such that
\begin{equation*}
	\opnorm{x_y - y} \le \frac{1}{\opnorm{\bG}}.
\end{equation*}
As $\opnorm{\bG} \le \epll{k}$, there exists a net $\calN$ of size $\exp(k \cdot \poly(\log\log(k)))$. We union bound over all the net vectors to obtain that with probability $\ge 99/100$, for all net vectors $x \in \calN$,
\begin{equation*}
	\opnorm{\bG x} \ge 2.
\end{equation*}
Now conditioning on this event, for an arbitrary $y \in \text{colspan}(A)$, $\opnorm{y}=1$, we have
\begin{align*}
	\opnorm{\bG y} &= \opnorm{\bG(x_y + (y-x_y))}\\
				   &\ge \opnorm{\bG x_y} - \opnorm{\bG(y-x_y)}\\
				   &\ge 2 - \opnorm{\bG}\opnorm{y-x_y}\\
				   &\ge 1
\end{align*}
as the net is chosen so that $\opnorm{y-x_y} \cdot \opnorm{\bG} \le 1$.

Conditioned on the event $\calE$, we have that each row of $\bG$ has at most $\epll{k}$ nonzero entries. Thus, each row of the matrix $\bG A$ can be computed in $k\cdot \epll{k}$ time and hence the matrix $\bG A$ can be computed in time $k^2\epll{k}$. As $\Pr[\calE] \ge 99/100$, the claim follows.
\end{proof}
\begin{theorem}[Subspace Embedding]
Given an ${n \times k}$ matrix $A$, we can compute an $m \times k$ matrix $\bS A$ with $m = k \cdot \poly(\log\log(k))$ such that with probability $\ge 9/10$, for all vectors $x \in \R^k$, 
\begin{equation*}
	\opnorm{Ax} \le \opnorm{\bS Ax} \le \epll{k}\opnorm{Ax}.
\end{equation*}
The matrix $\bS \cdot A$ can be computed in time $O(\nnz(A) + k^{2.1 + o(1)})$ or more generally in time $O(\gamma^{-1}\nnz(A) + k^{2+\gamma+o(1)})$ for any constant $\gamma > 0$. Further, for any matrix $M$ with $n$ rows, $$\E[\frnorm{\bS M}^2] \le \epll{k}\frnorm{M}^2.$$
\label{thm:sparse-subspace-embedding}
\end{theorem}
\begin{proof}
	The matrix $\bS$ is defined as follows
	\begin{equation*}
		\bS = 4 \cdot \bG \cdot \calF \cdot \bS_2 \cdot \bS_1
	\end{equation*}
	where $\bS_1$ is \textsf{OSNAP} for $k$ dimensional subspaces with $\gamma = 0.1$, $\bS_2$ is \textsf{OSNAP} for $k$ dimensional subspaces with $\gamma = 1/\log(k)$, $\calF$ is Indyk's embedding for $O(k\log(k))$ dimensional subspaces as in Theorem~\ref{thm:final-indyk-version} and $\bG$ is the sparse embedding matrix with $k \cdot \poly(\log\log(k))$ rows as in Theorem~\ref{thm:sparse-embedding}. We have with probability $\ge 9/10$, for any vector $x \in \R^k$,
	\begin{align*}
		\frac{1}{2}\opnorm{Ax} \le \opnorm{\bS_2 \cdot \bS_1 \cdot Ax} \le \frac{3}{2}\opnorm{Ax}.
	\end{align*}
	Condition on the above event. From Theorem~\ref{thm:final-indyk-version}, we have
	\begin{equation*}
		\frac14\opnorm{Ax} \le \frac{1}{2}\opnorm{\bS_2 \cdot \bS_1 \cdot A x} \le \opnorm{\calF \cdot \bS_2 \cdot \bS_1 \cdot Ax} \le \opnorm{\bS_2 \cdot \bS_1 \cdot Ax} \le \frac{3}{2}\opnorm{Ax}.
	\end{equation*}
	By Theorem~\ref{thm:final-indyk-version}, every unit vector in the span of $\calF$ has at least $Ck$ coordinates with an absolute value of at least $1/(\sqrt{k} \cdot \epll{k})$. Thus, the matrix $\calF \cdot \bS_2 \cdot \bS_1 \cdot A$ satisfies the conditions of Theorem~\ref{thm:sparse-embedding}. Therefore with probability $\ge 9/10$, we have for all vectors $x \in \R^k$,
	\begin{equation*}
		\opnorm{\bG \cdot \calF \cdot \bS_2 \cdot \bS_1 \cdot Ax} \le \epll{k}\opnorm{\calF \cdot \bS_2 \cdot \bS_1 \cdot Ax} \le \epll{k}\opnorm{Ax}
	\end{equation*}
	and
	\begin{equation*}
		\opnorm{\bG \cdot \calF \cdot \bS_2 \cdot \bS_1 \cdot Ax} \ge \opnorm{\calF \cdot \bS_2 \cdot \bS_1 \cdot Ax} \ge \frac{1}{4}\opnorm{Ax}.
	\end{equation*}
	Thus with probability $\ge 8/10$, for all vectors $x$, 
	\begin{equation*}
		\opnorm{Ax} \le \opnorm{\bS \cdot Ax} \le \epll{k}\opnorm{Ax}.
	\end{equation*}
	The matrix $\bS \cdot A$ can be computed as $4\bG(\calF(\bS_2(\bS_1A))))$ in time 
	\begin{equation*}
	O(\nnz(A) + k^{2.1}\log^2(k) + k^{2+o(1)} + k^2\cdot \epll{k})
	\end{equation*}
	where the last term follows from the fact that each of the $k\poly(\log\log(k))$ rows of the matrix $\bG$ has at most $\epll{k}$ nonzero entries.
	
	There is nothing special about $\gamma = 0.1$. We can choose any constant $1 > \gamma > 0$ and use \textsf{OSNAP} with the parameter $\gamma$ which gives an overall running time of $O({\gamma^{-1}}\nnz(A) + k^{2+\gamma + o(1)})$.
	
	We now bound $\E_{\bS}[\frnorm{\bS M}^2]$ for an arbitrary matrix $M$. We have
	\begin{align*}
		\E_{\bS}[\frnorm{\bS M}^2] &= 16\E_{\bG, \bS_2,\bS_1}[\frnorm{\bG \cdot \calF \cdot \bS_1 \cdot \bS_2 M}^2]\\
		&\le 16 \cdot \E_{\bS_1}[\E_{\bS_2}[\E_{\bG}[\frnorm{\bG \cdot \calF \cdot \bS_2 \cdot \bS_1M}^2\,|\, \bS_1, \bS_2]\,|\, \bS_1]].
	\end{align*}
	First, $\E_{\bG}[\frnorm{\bG \cdot \calF \cdot \bS_2 \cdot \bS_1M}^2\,|\, \bS_1, \bS_2] \le Mp \cdot (\text{scale}) \cdot \frnorm{\calF \cdot \bS_2 \cdot \bS_1M}^2$, where $M$ is the number of rows of $\bG$, $p$ is the probability of an entry of $\bG$ being nonzero and $\text{scale} = \epll{k}$ is the scaling factor for the random sign matrix. As $M = k \cdot \poly(\log\log(k))$ and $p = \epll{k}/k$, we have $\E_{\bG}[\frnorm{\bG \cdot \calF \cdot \bS_2 \cdot \bS_1M}^2\,|\, \bS_1, \bS_2] \le\epll{k} \cdot \frnorm{\calF \cdot \bS_2 \cdot \bS_1M}^2 \le \epll{k}\frnorm{\bS_2 \cdot \bS_1 M}^2$ as the matrix $\calF$ does not increase the Euclidean norm of any vector. Thus,
	\begin{equation*}
		\E_{\bS}[\frnorm{\bS M}^2] \le \epll{k}\E_{\bS_1}[\E_{\bS_2}[\frnorm{\bS_2 \cdot \bS_1 M}^2\,|\, \bS_1]] \le \epll{k}\frnorm{M}^2,
	\end{equation*}
	where the last inequality follows from the fact that $\frnorm{\bS_iM}^2$ is an unbiased estimator to $\frnorm{M}^2$ if $\bS_i$ is an \textsf{OSNAP}.
\end{proof}
\section{Applications}\label{sec:applications}

\subsection{Subspace Embeddings}
\begin{algorithm2e}
\caption{\textsc{LeverageScoreSampling}}
\label{alg:leverage-score-sampling}
	\KwIn{$A \in \R^{n \times k}$, $\varepsilon, \gamma > 0$}
	\KwOut{An $\varepsilon$ subspace embedding $\bS_{\textnormal{lev}} A$}
	\DontPrintSemicolon
	$\bS A \gets \textsc{SparseEmbedding}(A)$\;
	$[Q, R^{-1}] \gets \textsc{QR-Decomposition}(\bS A)$ \tcp*{$QR^{-1} = \bS A$}
	$s \gets k\exp(\poly(\log\log k)/\varepsilon^2$\; 
	$\bS_{1} \subseteq [n], f_i$ for $i \in [\bS_1] \gets$  \textsc{SampleFromProduct}($A, R, s,\gamma$)\tcp*{Lemma~\ref{lma:sampling-from-product}}
	For $i \in \bS_1$, set $(\bS_{\lev})_{ii}$ to be equal to $1/\sqrt{f_i}$\;
	\Return{$\bS_{\lev}A$ after removing $0$-value rows}
\end{algorithm2e}

We use the fast subspace embedding construction from previous sections to compute approximate leverage scores and then sample rows using the approximate leverage scores to compute $1+\varepsilon$ subspace embeddings in time $O(\gamma^{-1}\nnz(A) + \varepsilon^{-3}n^{\gamma}k^{2+o(1)}+k^{\omega}\poly(\log\log(k)))$ for any constant $\gamma$. We then compose with an \textsf{OSNAP} to obtain a subspace embedding with $O(\varepsilon^{-2}k\log(k))$ rows.

\begin{theorem}[Leverage Score Sampling]\label{thm:leverage-score-sampling}
Given a full column rank matrix $A \in \R^{n \times k}$, let $\ell_i^2$ for $i \in [n]$ be the leverage score of the $i$-th row. Let $p \in [0,1]^n$ be a vector of probabilities such that for all $i \in [n]$, $\min(1, r \cdot (\ell_i^2/k)) \ge p_i \ge \min(1, r \cdot \beta \cdot (\ell_i^2/k))$ for some $\beta < 1$, and let the $n \times n$ diagonal random matrix $\bS_{\lev}$ be defined as follows: for each $i \in [n]$, the entry $(\bS_{\lev})_{ii}$ is set to be equal to $1/\sqrt{p_i}$ with probability $p_i$, and is set to be $0$ with probability $1-p_i$. If $r \ge Ck\log(k)/\beta\varepsilon^2$ for an absolute constant $C$, then with probability $\ge 99/100$, for all vectors $x \in \R^{d}$
\begin{equation*}
    \opnorm{\bS_{\lev}Ax}^2 \in (1 \pm \varepsilon)\opnorm{Ax}^2.
\end{equation*}
With probability $\ge 1 - \exp(-\Theta(k))$, the matrix $\bS_{\lev}$ has at most $\Theta(Ck\log(k)/\beta\varepsilon^2)$ nonzero entries. 
\end{theorem}
The following lemma shows that a subspace embedding $S$ for the column space of a matrix $A$ can be used to compute approximate leverage scores which can be used to perform leverage score sampling as described above to obtain a $1+\varepsilon$ subspace embedding.
\begin{lemma}
\label{lma:embedding-to-apx-leverage-scores}
    If $S$ is a $\beta$ subspace embedding for the column space of a full rank matrix $A \in \R^{n \times k}$ i.e., for any vector $x$,
    \begin{equation*}
        {\opnorm{Ax}} \le \opnorm{S Ax} \le \beta\opnorm{Ax}
    \end{equation*}
    and if $SA = QR^{-1}$ for an orthonormal matrix $Q$, then for all $i \in [n]$,
    \begin{equation*}
        \ell_i^2/\beta^2 \le \opnorm{A_{i*}R}^2 \le \ell_i^2,
    \end{equation*}
    where $\ell_i$ is the leverage score of the $i$-th row of $A$.
\end{lemma}

The proof of the lemma is in Appendix~\ref{apx:proof-embedding-to-apx-leverage-scores}. Using our fast subspace embedding with $k\poly(\log\log(k))$ rows and $\beta = \epll{k}$, the above lemma shows that if we can compute the values $\opnorm{A_{i*}R}^2$, then we can obtain a $1+\varepsilon$ subspace embedding with $k\cdot \epll{k}/\varepsilon^2$ rows.

Often, the row norms $\opnorm{A_{i*}R}^2$ are approximated with $\opnorm{A_{i*}R\bG}^2$, where $\bG$ is a Gaussian matrix with $O(\log n)$ columns using the fact that for an arbitrary vector $x$, $\opnorm{\T{x}\bG}^2 \in (1/2,2)\opnorm{x}^2$ with probability $1 - 1/\poly(n)$. However, computing the matrix $AR\bG$ takes $O((\nnz(A)+k^2)\log(n))$ time. 

The following simple lemma shows that instead of obtaining constant approximations to $\opnorm{A_{i*}R}^2$ for all the rows by using a Gaussian matrix $\bG$ with $O(\log(n))$ columns, we can use a Gaussian matrix $\bG'$ with only $O(1/\gamma)$ columns to obtain $O(n^\gamma \log(n))$ factor approximations to $\opnorm{A_{i*}R}^2$. We sample the rows using these coarse approximations and then compute constant-factor approximations to $\opnorm{A_{i*}R}^2$ only for the rows that are sampled in the first stage and then reject each of the sampled rows with appropriate probabilities to obtain a leverage score sample.

\begin{lemma}\label{lma:sampling-from-product}
    Let $A \in \R^{n \times d}$ and $R \in \R^{d \times d}$ be such that for any vector $x \in \R^{d}$, the matrix-vector products $ARx, Rx$ can be computed in time at most $T_1$ and $T_2$ respectively. Given parameters $\gamma$ and $s$, there is an algorithm conditioned on an event $\calE$, $\Pr[\calE] \ge 95/100$, that samples indices $i \in [n]$ to obtain a random subset $\bS \subseteq [n]$, such that each $i \in [n]$ is in the set $\bS$ independently with probability $f_i$, where
\begin{equation*}
\min(1, s\frac{\opnorm{A_{i*}R}^2}{\frnorm{AR}^2}) \ge    f_i \ge \min(1, (s/16)\frac{\opnorm{A_{i*}R}^2}{\frnorm{AR}^2}).
\end{equation*}
The algorithm returns the random subset $\bS$ along with the probabilities $f_i$ for $i \in \bS$. The algorithm runs in time $O(\gamma^{-1}T_1 + T_2\log(n) + sdn^{\gamma}\log^2(n))$.
\end{lemma}
\begin{proof}
Let $p_i := \opnorm{A_{i*}R}^2/\frnorm{AR}^2$ for $i \in [n]$. Let $\bG_1$ be a Gaussian matrix with $O(1)$ rows and $n$ columns and $\bG_2$ be a Gaussian matrix with $d$ rows and $O(1)$ columns. We have
\begin{equation*}
\frac12\frnorm{AR}^2 \le \frnorm{\bG_1 AR \bG_2}^2 \le 2\frnorm{AR}^2 \quad {(\text{Event }\calE_1)}
\end{equation*}
with probability $\ge 99/100$. The matrix $\bG_1 AR\bG_2$ can be computed in $O(T_1 + n)$ time. Let $\bG_3$ be a Gaussian matrix with $O(\log(n))$ columns. With probability $\ge 99/100$,
\begin{equation*}
	\text{for all $i \in [n],$}\quad \frac{1}{2}\opnorm{A_{i*}R}^2 \le \opnorm{A_{i*}R\bG_3}^2 \le 2\opnorm{A_{i*}R}^2 \quad \text{(Event $\calE_2$)}.
\end{equation*}
We note that we \emph{do not} compute the matrix $AR\bG_3$ but we only compute the matrix $R\bG_3$ which can be done in time $O(T_2\log(n))$.

Now, let $\bG_4$ be a Gaussian matrix with $t = O(1/\gamma)$ columns. Let $\bg_1, \bg_2,\ldots, \bg_t$ be the columns of the matrix $\bG_4$. For each $i \in [n]$, with probability $\ge 1-1/100n^2$, $\max_{j \in [t]} |\langle A_{i*}R, \bg_j \rangle| \ge \opnorm{A_{i*}R}/n^{\gamma/2}$ using the fact that $|\langle A_{i*}R, \bg_j \rangle|_{j \in [t]}$ are independent half-Gaussians with standard deviation $\opnorm{A_{i*}R}$. By a union bound, with probability $\ge 1 - 1/100n$, for all $i \in [n]$, we have $\opnorm{A_{i*}R\bG_4}^2 \ge \max_{j \in [t]}\langle A_{i*}R, \bg_j\rangle^2 \ge \opnorm{A_{i*}R}^2/n^{\gamma}$. By Lemma~1 of \cite{lm2000}, we also obtain that with probability $\ge 1 - 1/100n$, for all $i \in [n]$, $\opnorm{A_{i*}R\bG_4}^2 \le O(\log(n))\opnorm{A_{i*}R}^2$. Thus, with probability $\ge 1 - 2/100n$, for all $i \in [n]$:
\begin{equation*}
	\frac{\opnorm{A_{i*}R}^2}{n^\gamma} \le \opnorm{A_{i*}R\bG_4}^2 \le C\log(n)\opnorm{A_{i*}R}^2 \quad \text{(Event $\calE_3$)}.
\end{equation*}
We compute $AR\bG_4$ and all squared row norms $\opnorm{A_{i*}R\bG_4}^2$ in time $O(T_1\gamma^{-1})$. Condition on the event $\calE := \calE_1 \cap \calE_2 \cap \calE_3$. We have $\Pr[\calE] \ge 95/100$.

Define $z_i := 2n^{\gamma}\opnorm{A_{i*}R\bG_4}^2/\frnorm{\bG_1AR\bG_2}^2$. We have $4Cn^{\gamma}\log(n)p_i \ge z_i \ge p_i$ and define $q_i := \min(1, sz_i)$. Sample $i \in [n]$ independently, each with probability $q_i$ to obtain a random subset $\bS_1 \subseteq [n]$. If $i \in \bS_1$, compute the value $\opnorm{A_{i*}(R\bG_3)}^2$ in time $O(d\log(n))$ and reject $i$ with probability $1 - \min(1, (s/4)\opnorm{A_{i*}R\bG_3}^2/\frnorm{\bG_1 AR\bG_2}^2)/q_i$. 

We need to show that this procedure is well-defined. We have $(s/4)\opnorm{A_{i*}R\bG_{3}}^2/\opnorm{\bG_1 AR\bG_2}^2 \le (s/4)(4p_i) = sp_i \le sz_i$ which implies that $\min(1, (s/4)\opnorm{A_{i*}R\bG_3}^2/\frnorm{\bG_1 AR\bG_2}^2) \le q_i$ and therefore the  rejection probability as defined is valid. Let $\bS_2$ be the subset obtained after performing the rejection step on $\bS_1$. The probability that a row $i \in \bS_2$ is 
\begin{equation*}
	f_i = q_i \cdot \frac{\min(1, (s/4)\opnorm{A_{i*}R\bG_3}^2/\frnorm{\bG_1 AR\bG_2}^2)}{q_i} \ge \min(1, (s/4)(p_i/4)) = \min(1, (s/16)p_i).
\end{equation*}
We also have that $f_i \le \min(1, sp_i)$. Thus with probability $\exp(-s)$ only $O(s)$ rows survive the rejection.

Now, with probability $\ge 1 - \exp(-s)$, $|\bS_1| = O(\sum_i q_i) = O(sn^{\gamma}\log(n))$  and therefore the squared row norm $\opnorm{A_{i}R\bG_3}^2$ has to be computed only for $O(sn^{\gamma}\log(n))$ rows. Thus, the time complexity of sampling  is $O(\gamma^{-1}T_1 + T_2\log(n) + O(sdn^{\gamma}\log^2(n)))$. Therefore, conditioned on the event $\calE$, the algorithm returns a subset $\bS \subseteq [n]$ sampled from the desired probability distribution in time $O(\gamma^{-1}T_1 + T_2\log(n) + sdn^{\gamma}\log^2(n))$.
\end{proof}

Using these lemmas, the following theorem shows that  Algorithm~\ref{alg:leverage-score-sampling} gives a $1+\varepsilon$ subspace embedding by sampling using approximate leverage scores.
\begin{theorem}
\label{thm:final-epsilon-subspace-embedding}
	Given a full rank matrix $A \in \R^{n \times k}$, a constant $\gamma$ and a parameter $\varepsilon > 0$, we have the following:
	\begin{enumerate}
		\item 	Algorithm~\ref{alg:leverage-score-sampling} computes a matrix $\bS_{\textnormal{lev}} A$ with $\Theta(\varepsilon^{-2}k\cdot \epll{k})$ rows such that with probability $\ge 9/10$, for all vectors $x$,
		\begin{equation*}
			\opnorm{\bS_{\textnormal{lev}} Ax}^2 \in (1 \pm \varepsilon)\opnorm{Ax}^2.
		\end{equation*}
		This matrix $\bS_{\textnormal{lev}} A$ can be computed in time 
		\begin{equation*}
		O(\gamma^{-1}\nnz(A)+ \varepsilon^{-2}n^{\gamma}k^{2+ o(1)}+k^{\omega}\poly(\log\log(k))).			
		\end{equation*}
		\item Composing $\bS_{\textnormal{lev}}$ with the matrix $\bS_{\textsf{OSNAP}}$, an \textsf{OSNAP} with $O(\varepsilon^{-2}k\log(k))$ and at most $O(\varepsilon^{-1}\log(k))$ nonzero entries in each column, we obtain that with probability $\ge 9/10$, for all vectors $x$,
		\begin{equation*}
			\opnorm{\bS_{\textsf{OSNAP}} \cdot \bS_{\textnormal{lev}} \cdot Ax}^2 \in (1 \pm O(\varepsilon))\opnorm{Ax}^2.
		\end{equation*}
		The matrix $\bS_{\textsf{OSNAP}} \cdot (\bS_{\textnormal{lev}}A)$ can be computed in time $O(\varepsilon^{-3}k^{2+o(1)})$ and hence, overall, the matrix $\bS_{\textsf{OSNAP}} \cdot \bS_{\textnormal{lev}} \cdot A$ can be computed in time 
		\begin{equation*}
			O(\gamma^{-1}\nnz(A)+k^{\omega}\poly(\log\log(k)) + \varepsilon^{-3}k^{2+o(1)} + \varepsilon^{-2}n^{\gamma+o(1)}k^{2 + o(1)})
		\end{equation*}
		for any constant $\gamma$.
	\end{enumerate}
\end{theorem}
\begin{proof}
From Theorem~\ref{thm:sparse-embedding}, we have a subspace embedding $\bS_{\fast}$ with $O(k\poly(\log\log k))$ rows and distortion $\epll{k}$  that can be applied to matrix $A$ in time $O(\gamma^{-1}\nnz(A) + k^{2+\gamma+o(1)})$ for any constant $\gamma > 0$. Compute the matrices $Q, R^{-1}$ such that $Q$ has orthonormal columns and $\bS_{\fast}A = QR^{-1}$ which can be done in time $O(k^{\omega}\poly(\log\log(k)))$. By Lemma~\ref{lma:embedding-to-apx-leverage-scores}, we have
\begin{equation*}
\frac{\ell_i^2}{\epll{k}} \le \opnorm{A_{i*}R}^2 \le \ell_i^2
\end{equation*}
which implies, using the fact $\sum_i \ell_i^2 = k$, that
\begin{equation*}
\frac{\ell_i^2}{k\cdot \epll{k}}	\le \frac{\opnorm{A_{i*}R}^2}{\frnorm{AR}^2}.
\end{equation*}
Using Lemma~\ref{lma:sampling-from-product}, conditioned on the event $\calE$, we can sample a random subset $\bS$ along with probabilities $f_i$ for $i \in \bS$ such that each $i \in [n]$ is independently in the subset $\bS$ with probability $f_i$,
\begin{equation*}
	f_i \ge \min(1, (s/4) \cdot \frac{\opnorm{A_{i*}R}^2}{\frnorm{AR}^2}) \ge \min(1, (s/4) \cdot \frac{\ell_i^2}{k\cdot \epll{k}}).
\end{equation*}
For $s = \Theta(k\log(k)\exp(\poly(\log \log k))/\varepsilon^2)$, we have $f_i \ge \min (1, C\ell_i^2\log(k)/\varepsilon^2)$ which implies that the matrix $\bS_{\lev}$ constructed by Algorithm~\ref{alg:leverage-score-sampling} is a $1+\varepsilon$ subspace embedding, with probability $\ge 9/10$, for the column space of $A$ by Theorem~\ref{thm:leverage-score-sampling}. In the notation of Lemma~\ref{lma:sampling-from-product}, for the matrices $A$ and $R$, $T_1 = \nnz(A) + k^2$ and $T_2 = k^2$. Thus, the sampling process runs in time 
\begin{align*}
	O(\gamma^{-1}\nnz(A) + k^{2}\log(n) + \varepsilon^{-2}n^{\gamma}k^2\exp(\poly(\log \log k))) &= O(\gamma^{-1}\nnz(A) + \varepsilon^{-2}n^{\gamma+o(1)}k^{2+o(1)}).
\end{align*}
Thus, overall, in time $O(\gamma^{-1}\nnz(A) + \varepsilon^{-2}n^{\gamma+o(1)}k^{2+ o(1)} + k^{\omega}\poly(\log\log k))$, we can compute a leverage score sampling matrix $\bS_{\lev}$ with $O(\varepsilon^{-2}k\exp(\log \log k))$ rows such that for all $x \in \R^k$,
\begin{equation*}
	\opnorm{\bS_{\lev}Ax}^2 \in (1 \pm \varepsilon)\opnorm{Ax}^2.
\end{equation*}

As $\nnz(\bS_{\lev}A) \le (\varepsilon^{-1}k)^2\epll{k}$, the \textsf{OSNAP} embedding $\bS_{\textsf{OSNAP}}$ can be applied to $\bS_{\lev}A$ in $O(\varepsilon^{-3}k^2\epll{k})$ time and the fact that $\bS_{\textsf{OSNAP}} \cdot \bS_{\lev}$ is a subspace embedding follows from the composability. Thus, we can compute $\bS_{\textsf{OSNAP}} \cdot \bS_{\lev} \cdot A$ which has $O(\varepsilon^{-2}k\log k)$ rows in $O(\gamma^{-1}\nnz(A) + k^{\omega}\poly(\log\log k) + \varepsilon^{-3}k^{2+o(1)}+\varepsilon^{-2}n^{\gamma+o(1)}k^{2+o(1)})$ time.

\end{proof}

\subsection{Linear Regression}
Let $A \in \R^{n \times k}$ and $b \in \R^{n}$. By the linear regression problem $(A,b)$, we mean $\min_x \opnorm{Ax-b}$ and $\text{OPT}(A,b)$ denotes the optimum value of this problem. We prove the following theorem.
\begin{theorem}
\label{thm:linear-regression}
	Given a full-rank matrix $A \in \R^{n \times k}$ and $b \in \R^n$, we obtain a solution $x^*$ such that
	\begin{equation*}
		\opnorm{Ax^* - b} \le (1+\varepsilon)\OPT(A,b)
	\end{equation*}
	in time $O(\gamma^{-1}\nnz(A) + \varepsilon^{-3}n^{\gamma+o(1)}k^{2+o(1)} + k^{\omega}\poly(\log\log k))$ for any constant $\gamma.$
\end{theorem}
\begin{proof}
	We first find a $1+\varepsilon$ subspace embedding $\bS$ for the column space of $[A, b]$. From Theorem~\ref{thm:final-epsilon-subspace-embedding}, $\bS A$ and $\bS b$ can be computed in at most $O(\gamma^{-1}\nnz(A) + \varepsilon^{-3}n^{\gamma+o(1)}k^{2+o(1)}+k^{\omega}\poly(\log\log k))$ time. We can also compute a preconditioner $R$ using the fast subspace embedding from Theorem~\ref{thm:sparse-embedding} such that 
\begin{equation*}
	\kappa(AR) = \epll{k}
\end{equation*}
by first computing $\bS_{\text{fast}}A = QR^{-1}$ and then inverting $R^{-1}$ to obtain $R$. The matrix $R$ can be computed in time $O(\gamma^{-1}\nnz(A) + k^{2+\gamma+o(1)} + k^{\omega}\poly(\log\log(k)))$ for any constant $\gamma$. We also have that
\begin{equation*}
	\kappa(\bS AR) = \epll{k}. 
\end{equation*}
Let $x^*$ be a solution such that
$
	\opnorm{\bS ARx^* - \bS b} \le (1+\varepsilon)\min_{x}\opnorm{\bS ARx - \bS b}.
$
Then, we have
\begin{equation*}
	\opnorm{ARx^* - b} \le \frac{1}{1-\varepsilon}\opnorm{\bS ARx^* - \bS b} \le \frac{1+\varepsilon}{1-\varepsilon}\opnorm{\bS Ax_{\text{opt}} - \bS b} \le \frac{(1+\varepsilon)^2}{1-\varepsilon}\opnorm{Ax_{\text{opt}} - b}.
\end{equation*}
Thus, $Rx^*$ is a $1+O(\varepsilon)$ approximate solution for the linear regression problem $(A,b)$. Now, we focus on obtaining a $1+\varepsilon$ approximate solution for the regression problem $(\bS AR, \bS b)$. 

We first compute an approximate solution for the regression problem as follows: let $\bS_{\text{fast}}$ be the subspace embedding with $k \poly(\log\log(k))$ rows for the column space of $[A, b]$. Let $x^{(0)} = (\bS_{\text{fast}} A)^{+} (\bS_{\text{fast}} b)$. This solution can be computed in time $O(\nnz(A) + k^{2+\gamma+o(1)} + k^{\omega}\poly(\log\log(k)))$. Let $x_{\text{start}} = R^{-1}x^{(0)}$ which can also be computed in time $O(k^2)$. Now, we have
\begin{equation*}
	\opnorm{\bS ARx_{\text{start}} - \bS b} \le (1+\varepsilon)\opnorm{ARx_{\text{start}} - b} = (1+\varepsilon)\opnorm{Ax^{(0)} - b} \le (1+\varepsilon)\opnorm{\bS_{\text{fast}}Ax^{(0)} - \bS_{\text{fast}} b}.
\end{equation*}
Let $x_{\bS}$ be the optimal solution for the regression problem $(\bS A, \bS b)$. By optimality of $x^{(0)}$ for the regression problem $(\bS_{\text{fast}} A, \bS_{\text{fast}}b)$, we have
\begin{align*}
	\opnorm{\bS ARx_{\text{start}} - \bS b} &\le (1+\varepsilon)\opnorm{\bS_{\text{fast}} Ax^{(0)} - \bS_{\text{fast}}b}\\
	&\le (1+\varepsilon)\opnorm{\bS_{\text{fast}} Ax_{\bS} - \bS_{\text{fast}} b}\\
	&\le (1+\varepsilon) \cdot \epll{k} \cdot \opnorm{Ax_{\bS} - b}\\
	&\le \epll{k} \cdot \text{OPT}((\bS A, \bS b)).
\end{align*}
Thus, $x_{\text{start}}$ is an $\epll{k}$ approximate solution for the linear regression problem $(\bS AR, \bS b)$. Using the solution $x_{\text{start}}$, we can obtain a $1+\varepsilon$ approximate solution in $O(\epll{k}/\varepsilon)$ iterations of gradient descent where each iteration can be performed in time $O(k^2\log(k)/\varepsilon^2)$. Thus, overall, in time 
\begin{equation*}
	O(\gamma^{-1}\nnz(A) + \varepsilon^{-3}n^{\gamma+o(1)}k^{2+o(1)} + k^{\omega}\poly(\log\log k)),
\end{equation*}
we can compute a $1+O(\varepsilon)$ approximate solution for the linear regression problem $(A,b)$.
\end{proof}

\subsection{Rank Computation and Independent Row Selection}
We give an algorithm to compute a maximal set of independent rows of an $n \times n$ matrix $A$ of rank $k = n^{\Omega(1)}$ in time $O(\gamma^{-1}\nnz(A) + k^{2+\gamma+o(1)} + k^{\omega}\poly(\log\log(k)))$ for any constant $\gamma > 0$, improving upon the earlier running time of $O((\nnz(A) + k^{\omega})\log(k))$ from \citet{CKL13} for any constant $\omega > 2$. 
\begin{Definition}[Rank Preserving Sketches]
A distribution $\calS$ over $z_S \times n$ matrices  is a rank preserving sketch if there exists a constant $c$ such that for $\bS \sim \calS$, with high probability, for a given matrix $A \in \R^{n \times d}$, $\min{(\rank(\bS A), z_S/c)} = \min{(\rank(A), z_S/c)}$ i.e., multiplying $A$ with the matrix $\bS$ preserves the rank if $\rank(A) \le z_S/c$. 
\end{Definition}

\begin{theorem}[\cite{CKL13}]
	There are rank-preserving sketching distributions as above with $c = 11$ such that
\begin{itemize}
	\item $\bS A$ can be computed in $O(\nnz(A))$ time
	\item $\bS$ has at most $2$ nonzero entries in a column
	\item $\bS$ has at most $2n/z_S$ nonzero entries in a row
\end{itemize}
\label{thm:ckl-rank-preserving}
\end{theorem}
They use rank preserving sketches to give an algorithm to compute the rank of an arbitrary matrix and an algorithm to compute a maximal set of linearly independent rows of the matrix. 
\begin{theorem}[Theorem~2.6 of \cite{CKL13}]
    Let $A \in \R^{n \times d}$ be an arbitrary matrix with $n \ge d$. There is a randomized algorithm to compute $k = \rank(A)$ in time $O(\nnz(A) \log(k) + \min(k^{\omega} , k \cdot \nnz(A)))$ with failure probability at most $O(1/n^{1/3})$. There is also an algorithm to find $k$ linearly independent rows of the matrix $A$ in time $O((\nnz(A) + k^{\omega})\log(n))$ with failure probability at most $O(\log(n)/n^{1/3})$.
    \label{thm:ckl-rank-computation}
\end{theorem}
We show that the $\log(k)$ factor can be removed from the time required to compute the rank of the matrix.
\begin{theorem}[Rank computation] \label{thm rank est}
Given $A \in \R^{n \times d}$, let $k = \rank(A)$. Let $\omega$ be the matrix multiplication constant and assume $\omega > 2$. Consider two cases:
\begin{enumerate}
 	\item If $k \le \log(n)^{2/(\omega-2)}$, $k$ can be computed in time $O(\nnz(A) + \log(n)^{6/(\omega-2)}) = O(\nnz(A))$.
 	\item If $k \ge \log(n)^{2/(\omega-2)}$, $k$ can be computed using Algorithm~\ref{alg:rank} (\textsc{Rank}) in time $O(\nnz(A) + \min(k^{\omega}, k \cdot \nnz(A)))$.
 \end{enumerate}
 \end{theorem}
 \begin{proof}
 If $k \le \log(n)^{2/(\omega-2)}$, then we have rank preserving sketches $S,R$ such that $SAR$ can be computed in time $\nnz(A)$, $SAR$ is an $O(\log(n)^{2/(\omega-2)}) \times O(\log(n)^{2/(\omega-2)})$ matrix and $\rank(SAR) = \rank(A)$. Now the rank of $SAR$ can be computed in time $O(\log(n)^{6/(\omega-2)})$. Thus, $\rank(A)$ can be computed in time $O(\nnz(A) + \log(n)^{6/(\omega-2)})$.

 In the case of $k\ge \log(n)^{2/(\omega-2)}$, consider Algorithm~\ref{alg:rank}. As $z \ge \Theta(\sqrt{n/\log(n)})$, with failure probability at most $\Theta(\sqrt{\log(n)/n})$, the sketch $SAR$ is rank preserving. As $SAR$ is a $z \times z$ matrix, we have $\nnz(SAR) \le z^2 \le O(\nnz(A)/\log(n))$. So, the rank $k_1$ of $SAR$ can be computed in time $O(\nnz(SAR)\log(k_1) + \min(k_1^{\omega}, k_1\cdot \nnz(SAR))$ by Theorem~\ref{thm:ckl-rank-computation}. As $k_1 \le k$, we have that the rank $k_1$ can be computed in time $O(\nnz(A) + \min(k^{\omega}, k\cdot \nnz(A)))$.

 We now have two cases. In the case that $k_1 < (\nnz(A)/\log(n))^{1/2}$, as we have $$\min(\rank(A), (\nnz(A)/\log(n))^{1/2}) = \min(\rank(S_1AR_1), (\nnz(A)/\log(n))^{1/2}),$$ we obtain that $\rank(A) = \rank(SAR) = k_1$.

 If $(\nnz(A)/\log(n))^{1/2} \le k_1$, we have $k=\rank(A) \ge k_1 \ge (\nnz(A)/\log n)^{1/2}$ which shows that $\nnz(A)\log(n) \le k^2\log^2(n) \le k^{\omega}$ for any $\omega > 2$ and $k \ge \log(n)^{2/(\omega-2)}$. We can now compute $\rank(A)$ in time $O(\nnz(A)\log(k) + \min(k^{\omega}, k\cdot \nnz(A)))$ by Theorem~\ref{thm:ckl-rank-computation}. As $\nnz(A)\log(k) = O(\min(\nnz(A) \cdot k, k^{\omega}))$, we obtain that the running time is $O(\nnz(A) + \min(k^{\omega},k \cdot \nnz(A)))$.
 \end{proof}


 \begin{algorithm2e}[H]
 \caption{$\textsc{Rank}(A)$}
 \label{alg:rank}
 \DontPrintSemicolon
 \KwIn{$A \in \R^{n \times d}$, $\rank(A) \ge (\log(n))^{6/(\omega-2)}$}
 \KwOut{$k  := \rank(A)$}
 \tcp{{\tt CKL-RE}, the algorithm of Theorem~2.6 of \cite{CKL13}}
 $z\gets c \cdot (\nnz(A)/\log n)^{1/2}$ \tcp*{$c\ge 1$ is a constant}
 Generate rank-preserving sketches $S\in\R^{z\times n}$ and $\T{R} \in\R^{z \times d}$\;
 Compute $SAR$ \tcp*{using Theorem~\ref{thm:ckl-rank-preserving}}
 $k_1 \gets \rank(SAR)$ \tcp*{using {\tt CKL-RE}}
 \If{$k_1 < z/c$}{\Return{$k_1$}}
$k_2\gets \rank(A)$ \tcp*{using {\tt CKL-RE}}
 \Return{$k_2$}\;
 \end{algorithm2e}


We now describe an algorithm to compute $k$ linearly independent rows of a matrix $A \in \R^{n \times d}$ of rank $k$ in time $O(\nnz(A) + k^{\omega}\poly(\log\log(n)))$, replacing the $\log(n)$ factor in the running time of \cite{CKL13} with $\poly(\log\log(n))$. Thus for matrices $A$ with $k^{\omega-1} \le \nnz(A) \le k^{\omega}/\log(n)$, we can now compute the rank $k$ and a set of $k$ linearly independent rows in time $O(k^{\omega}\poly(\log\log(k)))$ instead of $O(k^{\omega}\log(k))$ time. 

Without loss of generality, using the rank-preserving sketch, we can assume that $d = ck$ for a constant $c$. The following lemma describes a reduction to a  sparse sub-matrix of $A$ which also has rank equal to $\rank(A)$.
\begin{algorithm2e}
\caption{\textsc{RowReduction}($A,k$)}
\KwIn{$A \in \R^{n \times ck}$, $\rank(A) = k$}
\KwOut{$A_Q \in \R^{m \times ck}, m \le (3n/11)k, \nnz(A_Q) \le \max((2/5)\nnz(A), \Theta(k^2)), \rank(A_Q) = k$}
\DontPrintSemicolon
$\bS \gets \R^{ck \times n}$ be a rank-preserving sketch\;
Compute $\bS A$\;
Compute $P \subseteq [ck]$, $|P| = k$ such that $(\bS A)_P$ has $k$ linearly independent rows\;
Let $Q \gets \set{i \in [m]\,|\, \bS_{ji} \ne 0\,\text{for some $j \in P$}}$\;
\Return{$A_Q$}\;
\end{algorithm2e}

\begin{algorithm2e}
\caption{\textsc{IndependentRows}($A, k$)}	
\label{alg:lin-independent}
\DontPrintSemicolon
\KwIn{$A \in \R^{n \times d}, \rank(A) = k$}
\KwOut{$A_Q \in \R^{k \times d}, \rank(A_Q) = k$}
$\bS \gets \R^{ck \times d}$ be a rank preserving sketch\;
$B \gets A\T\bS$\;
Compute $B'$ by applying \textsc{RowReduction}  $\Theta(\log\log(n))$ times\;
Compute $\bS_{\text{lev}}$, a leverage score subspace embedding for $B'$ using Theorem~\ref{thm:final-epsilon-subspace-embedding} with $\gamma = 1/\log(n)$ and $\varepsilon = 0.1$\;
Compute $B''$ with $O(k)$ rows by applying \textsc{RowReduction} to the matrix $\bS_{\text{lev}}A$,  $\Theta(\log\log(k))$ times\;
Compute $k$ linearly independent rows of $B''$ and return $A_Q$ corresponding to these $k$ rows\;
\end{algorithm2e}

\begin{lemma}
    Let $A \in \R^{n \times ck}$ be an arbitrary matrix of rank $k$. There is a submatrix $A_Q \in \R^{m \times ck}$ that can be computed in time $O(\nnz(A) + k^{\omega})$ such that
    \begin{itemize}
        \item $m = |Q| \le (3n/11)$,
        \item $\nnz(A_Q) \le  \max((2/5) \cdot \nnz(A), \Theta(k^2))$, and
        \item $\rank(A_Q) = k$.
    \end{itemize}
    \label{lma:row-reduction}
\end{lemma}
\begin{proof}
Let $\bS \in \R^{ck \times n}$ be a rank-preserving sketch for $c = 11$. We have $\rank(\bS A) = \rank(A) = k$ with probability $\ge 1 - O(1/k)$. Consider a set $L$ of $k$ linearly independent rows of the matrix $\bS A$ which can be determined in $O(k^\omega)$ time. Let $Q \subseteq [n]$ be the set of rows of $A$ that contribute to the construction of the submatrix $(\bS A)_{L}$ which implies that $k \ge \rank(A_Q) \ge \rank((\bS A)_L) = k$ and hence $\rank(A_Q) = k$. We therefore have that the sub-matrix $A_{Q}$ consists of $k$ linearly independent rows. The reduction $A \rightarrow A_Q$ can be performed in $O(\nnz(A) + k^\omega)$ time. As each row of the matrix $\bS$ has at most $2n/11k$ nonzero entries, we have $|Q| \le (2n/11k) \cdot k \le 2n/11$. We now bound $\nnz(A_Q)$. 
	
	Let $P \subseteq [ck]$ be an arbitrary subset of size $k$. We show that if $Q_P \subseteq [n]$ is the subset of rows of $A$ that contribute to the construction of the sub-matrix $(\bS A)_P$, then $\nnz(A_{Q_P}) \le (2/5) \cdot \nnz(A)$ with high probability. 
	
	Let $\bX_i$ be the random variable that indicates if $A_{i*}$ contributes to the construction of $(\bS A)_P$ i.e., if $i \in Q_P$. By inspecting the proof of Theorem~\ref{thm:ckl-rank-preserving}, we obtain that $\Pr[\bX_i = 0] = (1-1/c)^2$. Thus, for $c = 11$, we obtain that $\Pr[\bX_i = 1] = 1 - (1 - 1/11)^2 = 21/121$. We also note that the random variables $\bX_1,\ldots, \bX_n$ are negatively associated \cite{wajc}. Let $a_i$ denote the number of nonzero entries of the row $A_{i*}$ which implies that $\sum_i a_i = \nnz(A)$. Now, we have $\nnz(A_{Q_P}) = \sum_{i}a_i\bX_i$. Using the Chernoff-Hoeffding bound for negatively associated random variables \cite{dubhashi1996balls},
	\begin{equation*}
		\Pr[\nnz(A_{Q_{P}}) = \sum_i a_i \bX_i \ge \nnz(A) \cdot 21/121 + t] \le 2\exp\left(-\frac{2t^2}{\sum_i a_i^2}\right).
	\end{equation*}
	By a union bound over all $\binom{11k}{k} \le (11e)^k$ subsets $P$, we obtain that for a constant $C$,
	\begin{equation*}
		\Pr[\text{There is a subset $P \subseteq [11k], |P|=k$ with $\nnz(A_{Q_P}) \ge \nnz(A)/5 + t$}] \le 2\exp\left(Ck - \frac{2t^2}{\sum_i a_i^2}\right).
	\end{equation*}
	Now, we have $\sum_i a_i^2 \le \max_i a_i \cdot \sum_{i} a_i \le 11k \cdot (\nnz(A))$ since the matrix $A$ is assumed to have only $ck=11k$ columns. For $t \ge \Theta(k \sqrt{\nnz(A)})$, we obtain that with probability $\ge 1 - \exp(-\Theta(k))$, for all $P\subseteq [11k], |P|=k$, we have that $\nnz(A_{Q_P}) \le \nnz(A)/5 + t$. For $\nnz(A) \ge \Theta(k^2)$, we have $\nnz(A)/5 \ge \Theta(k\sqrt{\nnz(A)})$ which implies that for all $P$, $\nnz(A_{Q_P}) \le (2/5)\nnz(A)$. This, in particular, implies that for $M = Q_{L}$, that corresponds to the set of rows contributing to a linearly independent set of rows of $(\bS A)$, we have $\nnz(A_M) \le (2/5) \cdot \nnz(A)$ if $\nnz(A) \ge \Theta(k^2)$. 
\end{proof}

Recursively applying the above lemma, we obtain the following.
\begin{corollary}
   	Let $A \in \R^{n \times d}$ be an arbitrary matrix of rank $k$. There is a matrix $A' \in \R^{m \times ck}$ with either $\nnz(A') \le \nnz(A)/\log(n)$ or $\nnz(A') \le \Theta(k^2)$ such that 
	\begin{itemize}
		\item $\rank(A') = \rank(A) = k$, and
		\item $m \le n/\poly(\log(n))$
		\item linearly independent rows of $A'$ correspond to linearly independent rows of $A$.
	\end{itemize}
	The reduction $A \rightarrow A'$ can be performed in $O(\nnz(A) + k^{\omega}\log\log(n))$ time. 
\end{corollary}
\begin{proof}
	Let $N = \Theta(\log\log(n))$ and $A^{(0)} = A$. Starting with $i = 0$, we apply the above reduction $A^{(i)} \rightarrow A^{(i+1)}$ to obtain a matrix with $\nnz(A^{(i+1)}) \le (2/5)\cdot \nnz(A^{(i)})$. Then $$\nnz(A^{(N)}) \le \max((2/5)^N\nnz(A), \Theta(k^2)) \le \max(\nnz(A)/\log(n), \Theta(k^2)).$$ The time complexity is $O(\sum_{i=1}^{N} (\nnz(A^{(i)}) + k^{\omega})) = O(\nnz(A) + k^{\omega}\log\log(n))$.
\end{proof}

We have now reduced the general problem of computing $k$ linearly independent rows of a rank-$k$ $n \times d$ matrix $A$ to computing $k$ linearly independent rows of a rank-$k$ $m \times ck$ matrix $A'$ with $m \le n/\poly(\log(n))$ and $\nnz(A') \le O(\max(k^2, \nnz(A)/\log(n)))$. Using these reductions, we have the following theorem.

\begin{theorem}
\label{thm:independent_rows}
	Given an arbitrary matrix $A \in \R^{n \times d}$ of rank $k$, Algorithm~\ref{alg:lin-independent} computes a set of $k$ linearly independent rows of the matrix $A$ in time $O(\nnz(A)+k^{\omega}\poly(\log\log(n))+ k^{2+o(1)})$.
\end{theorem}
\begin{proof}
	Let $\bS \in \R^{ck \times d}$ be a rank preserving sketch which implies $\rank(A\T\bS) = \rank(A) = k$ with probability $1 - O(1/k)$. Condition on this event. Let $M \subseteq [n]$, $|M| = k$ be such that rows of the sub-matrix $(A\T\bS)_{M} = A_M\T\bS$ are linearly independent. Then, $k \ge \rank(A_M) \ge \rank(A_M \T\bS) = k$ which implies $\rank(A_M) = k$. Thus, we only have to find $k$ linearly independent rows of the $n \times ck$ matrix $B = A\T\bS$. We also have $\nnz(B) = O(\nnz(A))$. Using the above corollary, we can find an $m \times ck$ sub-matrix $B'$ such that $\rank(B') = k$, $\nnz(B') \le O(\max(\nnz(B)/\poly(\log(n)), \Theta(k^2))$ and $m = n/\poly(\log(n))$. 
	
	From Theorem~\ref{thm:final-epsilon-subspace-embedding}, using $\gamma = 1/\log(n)$, in time $O(\nnz(B')\log(n) + k^{\omega}\poly(\log\log(n))+k^{2+o(1)} + m\gamma^{-1}) = O(\nnz(A) + k^{\omega}\poly(\log\log(n)) + k^{2+o(1)})$, we can compute a row sampling matrix $\bS_{\text{lev}}$ that samples $O(k\cdot \epll{k})$ rows such that
	\begin{equation*}
		\opnorm{\bS_{\text{lev}}B'x}^2 \in (1 \pm 1/10)\opnorm{B'x}^2
	\end{equation*}
	for all vectors $x$. This, implies that the matrix $\bS_{\text{lev}}B'$ has rank $k$ and hence has $k$ linearly independent rows.

As $\bS_{\textnormal{lev}}$ is a leverage score sampling matrix, the rows of $\bS_{\textnormal{lev}} B'$ are multiples of rows of the matrix $B'$. Thus, a set of $k$ linearly independent rows of the matrix $\bS_{\textnormal{lev}} B'$ directly corresponds to a set of $k$ linearly independent rows of $B$ which corresponds to a set of $k$ linearly independent rows of the matrix $A$.

Applying the row reduction $\poly(\log\log(k))$ times to the matrix $\bS_{\textnormal{lev}} B'$, we obtain a matrix $B''$ of dimension $O(k) \times k$ from which we can determine a set of $k$ linearly independent rows in time $O(k^{\omega})$. This concludes the proof. \end{proof}
\subsection{Low-Rank Approximation}
Let $A \in \R^{n \times d}$ be an arbitrary matrix. We want to compute a matrix $B$ of rank at most $k$ such that
\begin{equation*}
	\frnorm{A - B}^2 \le (1+\varepsilon)\frnorm{A - [A]_k}^2.
\end{equation*}
Let $\text{OPT}_A$ denote $\frnorm{A - [A]_k}^2$. Our main theorem for Low-Rank Approximation (LRA) is as follows.
\begin{theorem}
	Let $A \in \R^{n \times d}$, $k < \min(n,d)$ be a rank parameter and $\varepsilon > 0$ be an accuracy parameter. There is an algorithm that outputs matrices $V \in \R^{n \times k}$ and $X \in \R^{k \times d}$, $\T{V}V = I_k$, such that with $\Omega(1)$ probability,
	\begin{equation*}
		\frnorm{A - VX}^2 \le (1+\varepsilon)\frnorm{A - [A]_k}^2.
	\end{equation*}
	The algorithm runs in time $O(\gamma^{-1}\nnz(A) + {\varepsilon^{-1}}(n+d)k^{\omega-1} + \varepsilon^{-1}k(nd^{\gamma+o(1)}+dn^{\gamma+o(1)}) + \poly(\varepsilon^{-1}k))$ for any constant $\gamma > 0$.
	\label{thm:final-low-rank-approximation}
\end{theorem}
In the following sections, we will describe how to compute the left factor $V$ and the right factor $X$. We are not very careful with probabilities, as we only have to condition over the success of $O(1)$ events, and all these events can be chosen to have a success probability $1-c$ for any absolute constant $c > 0$ without affecting the time complexity.

We start with a residual sampling algorithm that lets us obtain a subspace containing a $1+\varepsilon$ approximation given a subspace that is only $O(1)$ approximate.
\subsubsection{Residual Sampling}
Suppose we have a subspace $V \in \R^{d}$ such that
\begin{equation*}
	\frnorm{A - A\mathbb{P}_V}^2 \le K\frnorm{A - [A]_k}^2.
\end{equation*}
The following theorem of \cite{deshpande2006matrix} shows that sampling $O(K \cdot k/\varepsilon)$ rows of the matrix $A$ with probabilities proportional to the squared distances of the rows to the subspace $V$ gives a subspace that along with $V$ contains a $1+\varepsilon$ rank-$k$ approximation to the matrix $A$.

\begin{theorem}[Theorem~2.1 of \cite{deshpande2006matrix}]
Let $A \in \R^{n \times d}$ and $V \in \R^{d}$ be a subspace. Let $E = A - A\mathbb{P}_V$, the matrix formed by projecting each row of $A$ away from the subspace $V$. Let $\bS$ be a random sample of $s$ rows of $A$ from a distribution $\calD$ such that row $i$ is chosen with probability $p_i \ge  \alpha\opnorm{E_{i*}}^2/\frnorm{E}^2$. Then for any non-negative integer $k$,
\begin{equation*}
	\E_{\bS}[\min_{\substack{\rank\text{-}k\, B\\
	\text{rowspan}(B) \subseteq V+\text{rowspan}(A_{\bS})}} \frnorm{A - B}^2] \le \frnorm{A - A_k}^2 + \frac{k}{s\alpha}\frnorm{E}^2.
\end{equation*}
\end{theorem}
Instead of sampling $s$ rows independently from the distribution $p$, we can also sample each $i \in [n]$ with probability $q_i := \min(1, sp_i)$ and obtain the same result for the resulting random subset of rows. Sampling each $i \in [n]$ independently with probability  $q_i$ lets us use the sampling framework from Lemma~\ref{lma:sampling-from-product}.
\begin{lemma}[Sampling each row independently]
	Let $A \in \R^{n \times d}$ and $V$ be a subspace in $\R^{d}$ and let $E = A - A\mathbb{P}_V$. Sample each $i \in [n]$ independently with a probability $q_i := \min(1, sp_i)$, with $p_i \ge \alpha \opnorm{E_{i*}}^2/\frnorm{E}^2$ to obtain a random subset $\bS \subseteq [n]$. For any nonnegative integer $k$,
	\begin{equation*}
	\E_{\bS}[\min_{\substack{\rank\text{-}k\, B\\
	\text{rowspan}(B) \subseteq V+\text{rowspan}(A_{\bS})}} \frnorm{A - B}^2] \le \frnorm{A - A_k}^2 + \frac{k}{s\alpha}\frnorm{E}^2.		
	\end{equation*}
	\label{lma:residual-sampling-our-version}
\end{lemma}
The proof of this lemma is in Appendix~\ref{apx:proof-residual-sampling}

\subsubsection{Computing the left factor of an approximation}
Let $\bT$ be a CountSketch matrix with $\Theta(k^2)$ columns. In \cite{cohen2015dimensionality}, the authors show that $\bT$ is a projection cost preserving sketch, i.e., with probability $9/10$, for all projection matrices $P$ of rank at most $O(k)$,
\begin{equation*}
	\frnorm{(I-P)A\bT}^2 = (1\pm 1/10)\frnorm{(I-P)A}^2.
\end{equation*}
Let $\bS$ be a CountSketch matrix with $\Theta(k^{4})$ rows. Then, with probability $\ge 99/100$, $\bS$ is a subspace embedding for the matrix $A\bT$ and therefore for any matrix $X$,
\begin{equation*}
	\frnorm{\bS A\bT X - \bS A\bT}^2 = (1 \pm 1/10)\frnorm{A\bT X - A\bT}^2.
\end{equation*}
We can relate $\text{OPT}_A$ and $\text{OPT}_{\bS A\bT}$ as follows: 
\begin{equation*}
\OPT_{\bS A\bT} = 	\frnorm{\bS A\bT - [\bS A\bT]_k}^2 = \min_{\text{rank-}k\, X}\frnorm{\bS A\bT - \bS A\bT X}^2 \le \frac{11}{10}\min_{\text{rank-}k\, X}\frnorm{A\bT - A\bT X}^2 = \frac{11}{10}\OPT_{A\bT}
\end{equation*}
where the inequality follows from the subspace embedding property of $\bS$ for the column space of $A\bT$. Now,
\begin{equation*}
	\OPT_{A\bT} = \min_{\text{\rank-$k$ projections}\, P}\frnorm{(I-P)A\bT}^2 \le \frac{10}{9}\min_{\text{\rank-$k$ projections}\, P}\frnorm{(I-P)A}^2 = \frac{10}{9}\OPT_A.
\end{equation*}
Here, the inequality follows as $\bT$ is a projection cost preserving sketch for $k$ dimensional projections. Thus, $\OPT_{\bS A\bT} \le (11/9)\OPT_{A}$.

\citet{boutsidis2017optimal} show that for any matrix $M$, there exists a sub-matrix $M'$ of $M$, with $O(k/\varepsilon)$ columns such that there is a rank $k$ matrix $B$, $\text{colspan}(B) \subseteq \text{colspan}(M')$, and $\frnorm{M - B}^2 \le (1+\varepsilon)\frnorm{M - [M]_k}^2$. They also give an algorithm to find such a subset of columns. As $\bS A\bT$ is a $O(k^4) \times O(k^2)$ matrix, using their algorithm, we can compute in time $\poly(k)$, a column selection matrix $\Omega$ that selects $O(k)$ columns of $\bS A\bT$ such that
\begin{equation*}
\min_{\text{rank-}k\, X}	\frnorm{\bS A\bT - \bS A\bT\Omega X}^2 \le \frac{3}{2}\OPT_{\bS A\bT} \le 2\OPT_A.
\end{equation*}
We now have
$
	\frnorm{(\bS A\bT \Omega)(\bS A\bT\Omega)^+\bS A\bT - \bS A\bT}^2 \le \min_{\text{rank-}k\, X}	\frnorm{\bS A\bT - \bS A\bT\Omega X}^2 \le 2\OPT_A.
$
Using the property that $\bS$ is a subspace embedding for the column space of $A\bT$, we have
\begin{equation*}
	\frnorm{A\bT \Omega(\bS A\bT)^+\bS A\bT - A\bT}^2 \le \frac{20}{11}\OPT_A.
\end{equation*}
Let $U$ be a matrix with orthonormal columns such that $\text{colspan}(A\bT \Omega) = \text{colspan}(U)$. Therefore,
\begin{equation*}
	\frnorm{U\T{U}A\bT - A\bT}^2 \le \frnorm{(A\bT \Omega)(\bS A\bT\Omega)^+\bS A\bT - A\bT}^2 \le \frac{20}{11}\OPT_A
\end{equation*}
which finally implies, as $\bT$ is a projection cost preserving sketch for $O(k)$ dimensional projections, that
$
	\frnorm{U\T{U}A - A}^2 \le ({10}/{9})({20}/{11})\OPT_A \le 3\OPT_A.
$
Thus, $\text{colspan}(U)$ is an $O(k)$ dimensional subspace with $\frnorm{(I-U\T{U})A}^2 \le 3\OPT_A$. As, $\bT$ and $\bS$ are CountSketch matrices, the matrices $A\bT$ and $\bS A\bT$ can be computed in time $\nnz(A)$. The matrix $\Omega$ can be computed in time $\poly(k)$ and the matrix $A\bT\Omega$ is obtained by selecting the appropriate columns of matrix $A\bT$. The orthonormal matrix $U$ can be computed in time $O(nk^{\omega-1})$. Using $U$, we now obtain a larger subspace of dimension $O(k/\varepsilon)$ that spans a $1+\varepsilon$ approximation.

Using Lemma~\ref{lma:residual-sampling-our-version}, we have that if \emph{columns} of the matrix $A$ are sampled independently to obtain a subset $\bS_{\text{res}}\subseteq [d]$ such that $\Pr[j \in \bS_{\text{res}}] \ge \min(1, sp_j)$ for $s = O(k/\varepsilon)$, $p_j = \opnorm{(I - U\T{U})A_{*j}}^2/\frnorm{(I-U\T{U})A}^2$, then with probability $\ge 99/100$, the subspace $\text{colspan}(U)+ \text{colspan}(A^{\bS_{\text{res}}})$ spans columns of a $k$ dimensional matrix that is a $(1+\varepsilon)$ rank-$k$ approximation for $A$.

Lemma~\ref{lma:sampling-from-product} shows how to sample $\bS_{\text{res}}$ from such a distribution. In the notation of  Lemma~\ref{lma:sampling-from-product}, we have $T_1 = O(\nnz(A) + nk)$ and $T_2 = nk$. Therefore, with probability $\ge 95/100$, we can obtain a sample $\bS_{\text{res}}$ from a distribution over subsets of $[d]$ such that independently, $\Pr[j \in \bS_{\text{res}}] \ge \min(1, O(k/\varepsilon)p_j)$ in time $O(\gamma^{-1}(\nnz(A) + nk) + nk\log(d) + \varepsilon^{-1}d^{\gamma}nk\log^2(d)) = O(\gamma^{-1}\nnz(A) + \varepsilon^{-1}nkd^{\gamma+o(1)})$ for any small constant $\gamma$. Let $M = [U\, A^{\bS_{\text{res}}}]$. We have with probability $\ge 9/10$, that
\begin{equation*}
	\min_{\text{rank-}k\, X}\frnorm{MX - A}^2 \le (1+\varepsilon)\OPT_A.
\end{equation*}
To obtain a good $k$-dimensional subspace within the column space of $M$, we can sketch and solve the above problem. Let $\bT_1$ be a CountSketch matrix with $O((k/\varepsilon)^2/\varepsilon^2)$ rows. Then with probability $\ge 99/100$, $\bT_1$ is an affine embedding for $(M, A)$ and therefore for any matrix $X$, $\frnorm{\bT_1 MX - \bT_1 A}^2 \in (1 \pm \varepsilon)\frnorm{MX - A}^2$. Let $X_{\bT_1}$ be the optimal solution for $\min_{{\text{rank-}}k\, X}\frnorm{\bT_1 MX - \bT_1 A}$. As $X_{\bT_1}$ is optimal, the rows of the matrix $X_{\bT_1}$ must be spanned by the rows of the matrix $\bT_1 A$, which implies that
$
	\min_{\text{rank-}k\, X}\frnorm{MX\bT_1 A - A}^2 \le (1+O(\varepsilon))\OPT_A.
$
This problem can now be solved by sketching on the left and the right with $\bT_1$ and $\bT_2$, where $\bT_2$ is a CountSketch matrix with $\poly(k/\varepsilon)$ rows, and then solving the sketched problem optimally. The time complexity of sketching is $O(\nnz(M) + \nnz(A)) = O(\nnz(A) + nk/\varepsilon)$, and the sketched problem can be solved in time $\poly(k/\varepsilon)$. Thus in time $O(\nnz(A) + nk/\varepsilon + \poly(k/\varepsilon))$, we can compute a rank $k$ matrix $X$ such that
\begin{equation*}
	\frnorm{MX\bT_1A - A}^2 \le (1+O(\varepsilon))\OPT_A.
\end{equation*}
We can also compute a decomposition of $X = X_1 \cdot X_2$ where $X_1$ has $k$ columns in time $\poly(k/\varepsilon)$, which implies that the $k$ dimensional column span of $MX_1$ is a $1+O(\varepsilon)$ approximate rank $k$ singular subspace i.e., $\frnorm{(MX_1)(MX_1)^+A - A}^2 \le (1+O(\varepsilon))\OPT_A$. The matrix $MX_1$ can be computed in time $O(nk^{\omega-1}/\varepsilon)$ and a matrix $V$ which is an orthonormal basis for the column space of  the $n \times k$ matrix $MX_1$ can be computed in time $O(nk^{\omega-1})$. Thus, in time $O(\gamma^{-1}\nnz(A) + \epsilon^{-1}nkd^{\gamma+o(1)} +\varepsilon^{-1}nk^{\omega-1} + \poly(\varepsilon^{-1}k))$, we can compute a left factor for a $1+\varepsilon$ rank-$k$ approximation of $A$. Thus, we have the following lemma.
\begin{lemma}
	Given a matrix $A \in \R^{n \times d}$, a rank parameter $k$ and accuracy parameter $\varepsilon$, we can compute a matrix $V$ with $k$ orthonormal columns in time $O(\gamma^{-1}\nnz(A) + \varepsilon^{-1}nkd^{\gamma+o(1)}+\varepsilon^{-(\omega-1)}nk^{\omega-1}+\poly(\varepsilon^{-1}k))$ such that
	\begin{equation*}
		\frnorm{A - V\T{V}A}^2 \le (1+\varepsilon)\frnorm{A - [A]_k}^2.
	\end{equation*}
\end{lemma}
\subsubsection{Computing a right factor given a left factor}
Given a matrix $V$ with $k$ orthonormal columns such that
\begin{equation*}
\min_{X}	\frnorm{VX - A}^2 \le (1+O(\varepsilon))\frnorm{A - [A]_k}^2,
\end{equation*}
we want to compute a rank $k$ matrix $\tilde{X}$ that satisfies $\frnorm{V\tilde{X} - A}^2 \le (1+O(\varepsilon))\frnorm{A - [A]_k}^2$.

For $i \in [n]$, let $p_i = \opnorm{V_{*i}}^2/k$. Suppose $\bS_{\lev}$ is a sampling matrix with $s = O(k\log(k))$ rows such that each row of $\bS_{\lev}$ is independently equal to $\T{e_i}/\sqrt{sp_i}$ with a probability $p_i$. Then we have 
\begin{equation*}
	\text{for all vectors $x$}, \opnorm{\bS_{\lev}Vx}^2 \in (1 \pm 1/2)\opnorm{Vx}^2.
\end{equation*}
Let $M_2 = \T{V}\T{\bS_{\lev}}$ and let $V_{M_2}$ be a matrix with $k$ orthonormal columns such that $\text{colspan}(V_{M_2}) = \text{rowspan}(M_2)$. Let $S_2$ be the BSS-Sampling matrix returned by the dual set spectral sparsification algorithm of \cite{boutsidis2017optimal} on the inputs $V_{M_2},\bS_{\lev}(I-V\T{V})A\bT$ with a parameter $4k$, where $\bT$ is a CountSketch matrix with $O(k^2)$ columns. The matrix $S_2$ selects $4k$ rows of the matrix $\bS_{\lev}A$. Let $R_1 = S_2\bS_{\lev}A$. Lemma~6.7 of \cite{boutsidis2017optimal} shows that
\begin{equation*}
	\frnorm{A - AR_1^+R_1}^2 \le O(1)\frnorm{A - [A]_k}^2.
\end{equation*}
As the matrix $R_1$ has $4k$ rows, an orthonormal basis $U$ for the rowspace of $R_1$, with $4k$ orthonormal columns, can be computed  in time $dk^{\omega-1}$. We can then perform residual sampling of rows of $A$ with respect to the subspace $U$ using the Lemma~\ref{lma:sampling-from-product}. Here $T_1 = \nnz(A)+dk$ and $T_2 = dk$. Thus, we can sample rows from a distribution defined by the probabilities $\min(1, (s/16){\opnorm{A_{i*}(I-U\T{U})}^2}/{\frnorm{A(I-U\T{U})}^2})$, for $s = O(k/\varepsilon)$ in time $O(\gamma^{-1}\nnz(A) + \varepsilon^{-1}dkn^{\gamma+o(1)})$. Let $\bS_{\text{res}}' \subseteq [n]$ be the rows sampled. Let $R = \begin{bmatrix}\T{U} \\ A_{\bS_{\text{res}}'}\end{bmatrix}$. The matrix $R$ has $O(k/\varepsilon)$ rows.

Now,  as in proof of the Theorem~5.1 of \cite{boutsidis2017optimal}, we have with proabability $\ge 9/10$,
\begin{equation*}
	\frnorm{A - V\T{V}AR^+R}^2 \le (1+O(\varepsilon))\frnorm{A - [A]_k}^2,
\end{equation*}
which implies
$
	\min_{X}\frnorm{A - VXR}^2 \le (1+O(\varepsilon))\frnorm{A - [A]_k}^2.
$
By sketching the problem on the left and the right with CountSketch matrices $\bT_1$ and $\bT_2$ with $\poly(k/\varepsilon)$ rows and columns respectively, the optimal solution $X_{\bT}$ for the sketched problem satisfies
\begin{equation*}
	\frnorm{A - VX_{\bT}R}^2 \le (1+O(\varepsilon))\frnorm{A - [A]_k}^2.
\end{equation*}
Finally, the product $X_{\bT} \cdot R$ can be computed in time $O(dk^{\omega-1}/\varepsilon) $ to obtain a matrix $\tilde{X}$ such that
\begin{equation*}
	\frnorm{A - V\tilde{X}}^2 \le (1+O(\varepsilon))\frnorm{A - [A]_k}^2.
\end{equation*}
Thus, we can compute two matrices $V,\tilde{X}$ with $k$ columns and $k$ rows respectively, such that the product $V \cdot \tilde{X}$ is a $1+\varepsilon$ approximate rank-$k$  Frobenius norm approximation to the matrix $A$, in time 
\begin{equation*}
	O(\gamma^{-1}\nnz(A) + \varepsilon^{-1}(n+d)k^{\omega-1} + \varepsilon^{-1}k(nd^{\gamma+o(1)}+dn^{\gamma+o(1)}) + \poly(\varepsilon^{-1}k)).
\end{equation*}

\paragraph{Acknowledgments:}  P. Kacham and D. Woodruff research were supported in part by National Institute of Health grant 5R01 HG 10798-2, NSF award CCF-1815840, and a Simons Investigator Award.

\bibliographystyle{plainnat}
\bibliography{main}
\appendix
\section{Missing proofs from Section~\ref{sec:applications}}\label{apx:missing-applications}
\subsection{Proof of Lemma~\ref{lma:embedding-to-apx-leverage-scores}}\label{apx:proof-embedding-to-apx-leverage-scores}
\begin{proof}[Proof of Lemma~\ref{lma:embedding-to-apx-leverage-scores}]
Let $AR = UT$ where $U$ is an orthonormal matrix. As $\text{colspan}(AR) = \text{colspan}(A)$, we have that $\ell_i^2 = \opnorm{U_{i*}}^2$. We first have for any vector $x$,
\begin{equation*}
	\opnorm{Tx} = \opnorm{UTx} = \opnorm{ARx} \le \opnorm{S ARx} = \opnorm{Qx} = \opnorm{x}
\end{equation*}
and
\begin{align*}
	\opnorm{Tx} = \opnorm{UTx} = \opnorm{ARx} \ge ({1}/{\beta})\opnorm{S ARx} &= ({1}/{\beta})\opnorm{Qx} =({1}/{\beta})\opnorm{x}.
\end{align*}
Here we repeatedly used the facts that $Q$ and $U$ are orthonormal matrices. Thus, we obtain $\opnorm{T} \le 1$ and $\sigma_{\min}(T) \ge 1/\beta$. As $A_{i*}R = U_{i*}T$, we obtain that
\begin{equation*}
	\opnorm{A_{i*}R}  = \opnorm{U_{i*}T} \le \opnorm{U_{i*}}\opnorm{T} \le \opnorm{U_{i*}}
\end{equation*}
and
\begin{equation*}
	\opnorm{A_{i*}R} = \opnorm{U_{i*}T} \ge \opnorm{U_{i*}}\simga_{\min}(T) \ge (1/\beta)\opnorm{U_{i*}}. 
\end{equation*}
Thus,
$
	{\ell_i}^2/{\beta^2} \le \opnorm{A_{i*}R}^2 \le \ell_i^2.
$
\end{proof}
\subsection{Proof of Lemma~\ref{lma:residual-sampling-our-version}}\label{apx:proof-residual-sampling}
\begin{proof}[Proof of Lemma~\ref{lma:residual-sampling-our-version}]
	Let $u^{(1)},\ldots, u^{(d)}$ be the left singular vectors and $v^{(1)},\ldots, v^{(d)}$ be the right singular vectors. For $j = 1,\ldots, k$, let
	\begin{equation*}
		\bX^{(j)} = \sum_{i: q_i < 1} \frac{u^{(j)}_i}{q_i}\T{(E_{i*})}\bI[\text{$i$ is sampled}]
	\end{equation*}
	and $\bw^{(j)} = \bX^{(j)} + \sum_{i:q_i=1}u_i^{(j)}\T{(E_{i*})}+ \mathbb{P}_V\T{A}u^{(j)}$. We have $\E[{\bw^{(j)}}] = \T{A}u^{(j)} = \sigma_jv^{(j)}$. Now,
	\begin{equation*}
		\E[\opnorm{\bw^{(j)} - \sigma_j v^{(j)}}^2] = \E[\opnorm{\bX^{(j)} - \sum_{i:q_i < 1}u^{(j)}_i\T{(E_{i*})}}^2] = \E[\opnorm{\bX^{(j)}}^2] - \opnorm{\sum_{i:q_i < 1}u^{(j)}_i\T{(E_{i*})}}^2.
	\end{equation*}
Now,
\begin{align*}
	\E[\opnorm{\bX^{(j)}}^2] &= \E[\opnorm{\sum_{i: q_i < 1}\frac{u^{(j)}_i}{q_i}\T{(E_{i*})}\bI[\text{$i$ is sampled}]}^2]\\
	&= \sum_{i:q_i < 1} \frac{(u_i^{(j)})^2}{q_i^2}\opnorm{E_{i*}}^2q_i + \sum_{i \ne i': q_i,q_{i'} < 1} u_{i}^{(j)}u^{(j)}_{i'}\langle E_{i*}, E_{{i'*}}\rangle\\
\end{align*}
As the values $p_i$ used to define probabilities $q_i$ are such that $p_i \ge \alpha \opnorm{E_{i*}}^2/\frnorm{E}^2$, then we have
\begin{equation*}
	\E[\opnorm{\bX^{(j)}}^2] \le \frac{1}{s\alpha}\frnorm{E}^2 + \opnorm{\sum_{i: q_i < 1} u_{i}^{(j)}\T{(E_{i*})}}^2 - \sum_{i: q_i < 1}\opnorm{u_{i}^{(j)}\T{(E_{i*})}}^2.
\end{equation*}
Thus, $\E[\opnorm{\bw^{(j)} - \sigma_j v^{(j)}}^2] \le (1/s\alpha)\frnorm{E}^2 - \sum_{i: q_i < 1}\opnorm{u_{i}^{(j)}\T{(E_{i*})}}^2$. From here, using the same proof as \cite{deshpande2006matrix}, we obtain that the subspace $V + \text{span}(A_S)$ spans rows of a rank $k$ matrix $B$ such that
\begin{equation*}
	\frnorm{A-B}^2 \le \frnorm{A-A_{k}}^2 + \frac{k}{s\alpha}\frnorm{E}^2.
\end{equation*}
\end{proof}

\end{document}